\documentclass[11pt]{article}

\usepackage[margin=1in]{geometry}
\usepackage{amsmath,amssymb,amsthm,mathtools}
\usepackage{keytheorems}
\usepackage{array,booktabs,tabularx}
\usepackage{float}
\usepackage{bbold}
\usepackage{mdframed}
\usepackage{bookmark}
\usepackage{hyperref}
\usepackage{orcidlink}
\usepackage[nameinlink,capitalise]{cleveref}

\usepackage[
    backend=biber,
    style=alphabetic,
    url=false,
    doi=false,
    isbn=false,
    sorting=nyt,
    sortcites=true,
    backref=true,
    maxbibnames=10,
    maxcitenames=3,
    mincitenames=1,
]{biblatex}
\usepackage{xcolor}

\hypersetup{
  hidelinks,
  colorlinks,
  linkcolor=[rgb]{0.3,0.3,0.6},
  citecolor=[rgb]{0.2,0.6,0.2},
  urlcolor=[rgb]{0.6,0.2,0.2}
}

\DefineBibliographyStrings{english}{%
  backrefpage  = {cited on page},
  backrefpages = {cited on pages}
}

\newcommand{\doiorurl}{%
  \iffieldundef{doi}
    {\iffieldundef{url}
       {}
       {\strfield{url}}}
    {http://dx.doi.org/\strfield{doi}}%
}

\newcommand{\myhref}[1]{%
  \ifboolexpr{%
    test {\ifhyperref}
    and
    not test {\iftoggle{bbx:url}}
    and
    not test {\iftoggle{bbx:doi}}
  }
    {\href{\doiorurl}{#1}}
    {#1}%
}

\DeclareFieldFormat{title}{\myhref{\mkbibemph{#1}}}
\DeclareFieldFormat
  [article,inbook,incollection,inproceedings,patent,thesis,unpublished]
  {title}{\myhref{\mkbibquote{#1\isdot}}}

\newenvironment{funding}
  {\par\bigskip\begingroup\small\noindent\textbf{Funding.}\ }
  {\par\endgroup}

\numberwithin{equation}{section}

\newkeytheorem{theorem}[numberlike=equation,style=plain]
\newkeytheorem{lemma}[numberlike=equation,style=plain]
\newkeytheorem{corollary}[numberlike=equation,style=plain]
\newkeytheorem{proposition}[numberlike=equation,style=plain]
\newkeytheorem{claim}[
  numberlike=equation,style=plain,
  refname={claim,claims},Refname={Claim,Claims}]

\newenvironment{claimproof}
  {\begin{proof}[Proof of the claim]}
  {\end{proof}}

\newkeytheorem{example}[numberlike=equation,style=plain]
\newkeytheorem{definition}[numberlike=equation,style=plain]
\newkeytheorem{observation}[
  numberlike=equation,style=plain,
  refname={observation,observations},Refname={Observation,Observations}]
\newkeytheorem{fact}[
  numberlike=equation,style=plain,
  refname={fact,facts},Refname={Fact,Facts}]
\newkeytheorem{remark}[numberlike=equation,style=remark]

\newcommand{\inbrace}[1]{\left\{#1\right\}}
\newcommand{\setdef}[2]{\inbrace{{#1}\ \mid\ {#2}}}

\newcommand{\F}{\mathbb{F}}
\newcommand{\N}{\mathbb{N}}
\newcommand{\Z}{\mathbb{Z}}
\newcommand{\R}{\mathbb{R}}
\newcommand{\C}{\mathbb{C}}

\newcommand{\rank}{\operatorname{rank}}
\newcommand{\Det}{\operatorname{\mathsf{Det}}}
\newcommand{\Perm}{\operatorname{\mathsf{Perm}}}
\newcommand{\HH}{\mathsf{H}}

\newcommand{\vecalpha}{\mathrm{\boldsymbol{\alpha}}}
\newcommand{\vecx}{\mathrm{\boldsymbol{x}}}
\newcommand{\vecy}{\mathrm{\boldsymbol{y}}}

\renewcommand{\char}{\textrm{char}}
\renewcommand{\epsilon}{\varepsilon}
\renewcommand{\tilde}{\widetilde}

\newcommand{\Blaser}{Bl{\"{a}}ser}

\newcommand{\IP}{\mathsf{IP}}
\newcommand{\EQ}{\mathsf{EQ}}
\newcommand{\UEQ}{\mathsf{UEQ}}
\newcommand{\SD}{\mathsf{SD}}
\newcommand{\SI}{\mathsf{SI}}
\newcommand{\BF}{\mathsf{BF}}
\newcommand{\HC}{\mathsf{HC}}
\newcommand{\ESym}{\operatorname{\mathsf{ESym}}}

\newcommand{\PCC}{\operatorname{PCC}}
\newcommand{\DCC}{\operatorname{DCC}}
\newcommand{\pep}{polynomial evaluation problem}
\newcommand{\Pep}{Polynomial evaluation problem}

\newcommand{\Asc}{\text{algebraic scanner}}
\newcommand{\APsc}{\text{algebraic probabilistic scanner}}
\newcommand{\ltAsc}{(\ell,t)\text{-scanner}}
\newcommand{\ltAPsc}{(\ell,t)\text{-probabilistic scanner}}

\definecolor{lightergray}{gray}{0.95}
\newenvironment*{problems}[3]
  {\begin{mdframed}[
     linecolor=lightergray,
     linewidth=2pt,
     leftmargin=0.5pt,
     rightline=false,
     bottomline=false,
     topline=false]
   \textbf{#1}:\newline
   \underline{Given:} #2\newline
   \underline{Check:} #3}
  {\end{mdframed}}

\title{Tight Lower Bounds for Algebraic Communication and Applications}
\author{\normalsize
  Manon Blanc\textsuperscript{1}\,\orcidlink{0000-0002-6961-089X}
  \quad
  Prateek Dwivedi\textsuperscript{1}\,\orcidlink{0000-0002-0572-3721}
  \quad
  Magnus Rahbek Dalgaard Hansen\textsuperscript{1}\,\orcidlink{0009-0002-5157-0268}
  \\[0.4em]
  Nutan Limaye\textsuperscript{1}\,\orcidlink{0000-0002-0238-1674}
  \qquad
  Meena Mahajan\textsuperscript{2}\,\orcidlink{0000-0002-9116-4398}}
\date{}

\begin{document}
\maketitle
\begingroup
\renewcommand{\thefootnote}{\arabic{footnote}}
\footnotetext[1]{IT University of Copenhagen, Copenhagen, Denmark.}
\footnotetext[2]{The Institute of Mathematical Sciences, Chennai, India; Homi Bhabha National Institute, Mumbai, India.}
\endgroup
\begin{abstract}
Communication complexity studies how much information must be exchanged to solve a problem whose input is split among two or more parties. The classical setting deals with Boolean inputs split between two parties. In this work, we study an algebraic variant, where the inputs are vectors over a field $\F \in \{\R, \C\}$.

There are two players, Alice and Bob, with inputs $ X\in \F^n$ and $Y\in \F^n$, respectively. We consider two kinds of tasks. In the \emph{polynomial evaluation problem}, the goal is to compute the value of a polynomial $ g\in \F[X,Y]$ on the given input. In the \emph{set-recognition problem}, the goal is to decide whether $ (X,Y)\in S $, for a given set $S\subseteq \F^{n} \times \F^n$. In both settings, Alice and Bob send evaluations of polynomials depending only on their own inputs. In the polynomial evaluation problem, these messages are combined to compute the evaluation $g(X,Y)$. In the set-recognition problem, a referee receives the messages and may apply polynomial tests to the messages received so far; the outcomes of these tests determine acceptance or rejection. The protocols may be deterministic or probabilistic.

This algebraic communication viewpoint, initiated by Abelson (JACM, 1980) and later developed by Grigoriev (Computational Complexity, 2008), has remained comparatively unexplored. In this work, we revisit this model and demonstrate that it continues to offer a rich framework for new lower bound questions. We present three sets of results.
\begin{itemize}
\item \textbf{Upper bounds and reductions.} We give non-trivial upper bounds for a range of natural polynomial evaluation and set-recognition problems, showing that algebraic communication can be significantly more powerful than simply sending all input coordinates.
We also prove reductions between different problems, which help organize the landscape of the model and identify which problems capture its main difficulties. These results serve two purposes: they illustrate the expressive power of the model, and they show that proving meaningful lower bounds in this setting is interesting.

\item \textbf{A lower bound framework and tight lower bounds.} Our main technical contribution is a general framework for proving lower bounds for algebraic set-recognition problems.
Using this framework, we prove several
probabilistic lower bounds for natural problems introduced earlier in the paper. In a number of cases, these lower bounds match the corresponding upper bounds, giving tight or near-tight characterizations of their algebraic communication.
 Along the way, we also recover and generalize some of Grigoriev's original lower bound results.

\item \textbf{Applications of the framework.} Finally, we give two applications of our framework. First, we use the communication lower bounds to prove lower bounds for a class of left-to-right algebraic algorithms, which we call algebraic scanners. Second, we show that our lower bound arguments extend beyond polynomials to a more general algebraic computational setting inspired by the Blum--Shub--Smale model.

\end{itemize}

\end{abstract}
\clearpage
\begin{funding}
MB, PD, MH, and NL acknowledge support from the Independent Research Fund Denmark (grant agreement No.\ 10.46540/3103-00116B) and from Basic Algorithms Research Copenhagen (BARC), funded by VILLUM Foundation Grant 54451. NL is supported by the Carlsberg Foundation grant CF25-1645.  MM acknowledges support from the J C Bose fellowship grant JCB/2023/000006 of the Anusandhan National Research Foundation (ANRF), India.
\end{funding}
\begingroup
\setcounter{tocdepth}{2}
\tableofcontents
\endgroup
\clearpage

\section{Introduction}
\label{sec:intro}

Communication complexity studies the amount of communication required to compute a function jointly by two or more parties on an input distributed among them. The standard setup has two players. The inputs are usually Boolean strings, and the two players aim to compute a Boolean function. This model was introduced by Yao~\cite{Yao79}. It is mathematically simple, and yet it captures many important aspects related to communication. It is extremely well-studied and has a large range of applications to other areas of algorithms and complexity theory such as circuit complexity, proof complexity, query complexity, and streaming algorithms~\cite{RaoYehudayoff,MuthukrishnanSurvey}.

A related model for a more general class of functions was introduced by Abelson~\cite{Abelson80} soon after Yao's two-party model for Boolean functions. Here, the two players, henceforth called Alice and Bob, receive inputs from $\F \in \{\R, \C\}$.
So, Alice receives $X = (X_1, \ldots, X_n) \in \F^n$ and Bob receives $Y = (Y_1, \ldots, Y_n) \in \F^n$. The goal is to evaluate a polynomial $g \in \F[X,Y]$ on the given input $(X,Y)$. We will call this the \emph{\pep}. The way they are allowed to do this is as follows: There is a referee receiving messages from Alice and Bob, say $a_1(X), \ldots, a_{r_1}(X) \in \F[X]$ and $b_1(Y), \ldots, b_{r_2}(Y) \in \F[Y]$, respectively. Now, the referee computes $g(X,Y) = P(a_1(X),\ldots, a_{r_1}(X), b_1(Y), \ldots, b_{r_2}(Y))$, where $P \in \F[z_1, \ldots, z_{r_1+r_2}]$. Among all such representations of $g$, we wish to find the one that minimizes $r_1+r_2$.  The problem has both computational and algebraic motivations. From the computational point of view, it is closely related in spirit to Yao’s communication model, since it studies computation under a partition of the input. From the algebraic point of view, it asks how succinctly can a polynomial in variables $X$ and $Y$ be represented using polynomials depending only on $X$ and polynomials depending only on $Y$.

More than two decades after this, Grigoriev~\cite{Gri2008} introduced a variant of the above model for decision problems in this setting. These problems are called set-recognition problems\footnote{Simultaneously, \Blaser\ and Vicari~\cite{BV2008} also introduced a closely related model. We will talk about it in detail in the section about related work.}. In this case, the goal is to recognize a set $S \subseteq \F^n \times \F^n$: given $X \in \F^n$ to Alice and $Y \in \F^n$ to Bob, check whether $(X,Y) \in S$ or not. Here, the protocol is visualized as a tree. Every node of the tree belongs to either Alice or Bob. At an Alice node, say $v$, she sends a message, a polynomial $a_v \in \F[X]$, to the referee. Similarly, at a Bob node, say $u$, Bob sends a polynomial $b_u(Y) \in \F[Y]$ to the referee. At each node, the referee applies \emph{a test}\footnote{In the original model~\cite{Gri2008}, each node can perform multiple tests. Here, we assume that each node has only $1$ test. However, all our results hold for the original model as well.} to the messages it has received so far. The tests are also polynomials (and not arbitrary functions). The tree branches based on the result of the test. For $\F = \C$, the protocol tests $\{=0, \neq 0\}$, whereas, for $\F = \R$, the tests are $\{<0, =0, >0\}$.
 The leaves of the protocol are labelled by outputs, either accept or reject; an input is accepted or rejected according to the label of the leaf reached by the protocol on that input. The depth of the tree is the complexity measure, which we want to minimize. A noteworthy aspect of Grigoriev's paper is that it also analyzes a probabilistic communication model. (Formal definitions appear in \cref{sec:prelims}.)
 We write $\DCC_{\F}(S)$ and $\PCC_{\F}(S)$ for the deterministic and probabilistic communication complexities of recognizing $S$, respectively, over the field $\F$.

 Our work focuses on Abelson and Grigoriev's models of communication. We study the polynomial evaluation problem as in Abelson's work, as well as the set-recognition problem, as in Grigoriev's work. We consider both deterministic and probabilistic communication.
 The set-recognition problem is a decision problem and, thus, bears some resemblance to decision problems arising in the communication complexity of Boolean functions, such as those studied in~\cite{KushilevitzNisan, RaoYehudayoff}. While Boolean communication complexity is often amenable to combinatorial techniques, our setting is inherently different: the domain is $\F^n \times \F^n$, rather than a discrete Boolean domain, and the communicated messages are evaluations of polynomials. This calls for a different set of tools, combining algebraic techniques with analytic reasoning.

In~\cite{Gri2008}, the first strong lower bounds for probabilistic communication were proved. In our work, we present a general framework for proving lower bounds for probabilistic communication complexity of set-recognition problems. This involves a substantial generalization of previous techniques~\cite{Gri2008}. The generalization also achieves a quantitative improvement to the lower bounds obtained in~\cite{Gri2008}. Furthermore, we give an extensive analysis of the deterministic model by providing new upper and lower bounds. We also give two applications of our results. The first application uses our lower bounds as a black-box to obtain lower bounds for algorithmic problems. This is similar to the standard communication-to-streaming connection. The second application involves modifying our lower bound proofs further in order to obtain lower bounds for a BSS-inspired communication model, which is more general than the communication model by~\cite{Gri2008}.

\paragraph*{\Pep.} As mentioned above, the problem was first studied in~\cite{Abelson80}. It is straightforward that the complexity of the problem is always upper bounded by $2n$. Alice and Bob can simply send their entire input to the referee and then the referee can compute $g(X,Y)$. Notice that Alice, Bob, and the referee are allowed to compute arbitrary polynomials of their inputs. That is, the complexity of computing $a_i$'s, $b_i$'s and $P$ does not contribute towards the cost of solving the problem.

The main contribution of Abelson's work was an elegant mathematical statement showing that the complexity of the problem is lower bounded by twice the rank of a matrix, which we call the mixed Hessian matrix and denote it by $\HH_{X|Y}(g)$. It is an $n\times n$ matrix, where the $(i,j)$th entry is ${\partial^2 g}/{\partial X_i \partial Y_j}$.

For example, consider the inner product polynomial $\IP_n(X,Y) = \sum_{i=1}^n X_i Y_i$. It is not hard to see that $\HH_{X|Y}(\IP_n) = I_n$, where $I_n$ is the $n\times n$ identity matrix. Thus, Abelson's bound~\cite{Abelson80} gives  a tight lower bound of $2n$ for the \pep\ for $\IP_n(X,Y)$.

\paragraph*{Set-recognition problem.}
Here is another instructive example of the equality polynomial. For this discussion, let us fix $\F = \R$. Let $\EQ_n(X,Y) = \sum_{i=1}^n (X_i-Y_i)^2$. Since a sum of squares over $\R$ vanishes exactly when every summand vanishes,
\[
V(\EQ_n)=\{(X,X)\mid X\in\R^n\}.
\]
The $2n$ deterministic upper bound is not hard to see: Alice sends $X_i$ to the referee, then Bob sends $Y_i$ to the referee, and the referee tests whether $(X_i-Y_i)^2 = 0$ or not. If it is equal to $0$, then the protocol proceeds to check the next coordinate else it rejects and terminates. In our work, we show that $2n$ is tight for deterministic protocols. (We will elaborate on this in \cref{sec:our-results}.\footnote{In~\cite{BV2008} a $2n$ lower bound was proved, but this was for a restricted setting where $a_i$'s and $b_i$'s are homogeneous. We do not need such a restriction.})

In contrast, one can recognize $V(\EQ_n)$ using a probabilistic protocol with $O(1)$ communication. Using public randomness, Alice and Bob choose $r$ uniformly from a fixed set $A\subseteq\R$ of size at least $3n$. Alice computes $a(X) = \sum_{i=1}^n X_i r^i$. Similarly, Bob computes $b(Y) = \sum_{i=1}^n Y_i r^i$. The referee accepts if and only if $a(X) = b(Y)$. As a univariate polynomial of degree at most $n$ can have at most $n$ roots, the probability that the protocol accepts when $X$ and $Y$ are not equal is at most $1/3$. This shows that the probabilistic protocols are provably stronger than deterministic protocols.
A similar exponential gap was proved in~\cite{Gri2008} for a set-recognition problem called the \emph{Orthant problem}. These gaps underline the fact that the model is non-trivial and proving lower bounds for probabilistic communication in this model is significant. The bounds for $V(\EQ_n)$ mentioned above may remind the reader of analogous bounds for the communication complexity of the Boolean equality function. We emphasize, however, that this similarity is only superficial: as we will see below, the techniques involved here are quite different.

In~\cite{Gri2008}, such lower bounds were achieved. In particular, it proved the following statement.

\begin{theorem}[\cite{Gri2008}]
\label{thm:grigoriev}
Let $\IP_n(X,Y) = \sum_{i=1}^n X_i Y_i$, let $S(\IP_n) = \{(X,Y) \in \R^n \times \R^n \mid \IP_n(X,Y) \leq 0\}$, and let $V(\IP_n) = \{(X,Y) \in \C^n \times \C^n \mid \IP_n(X,Y) = 0\}$. Then, $\PCC_{\R}(S(\IP_n))$ is at least $2n-6$ and  $\PCC_{\C}(V(\IP_n))$ is at least $2n-6$.
\end{theorem}

\begin{remark}
    The theorem above summarizes Proposition 3.1 and Corollary 3.2 of Grigoriev~\cite{Gri2008}. In fact, these results establish a slightly more general lower bound than the version stated here, and our work generalizes those original statements as well. For ease of exposition, we state only this restricted form in this part of the introduction. We return to the more general formulation later, when discussing the proof techniques underlying our results.
\end{remark}

\subsection{Our results and techniques}
\label{sec:our-results}

We start with a summary of our results and organization. After this brief summary, we discuss the main theorems and the proof techniques used to prove them.

\subsubsection{Summary and organization}

\begin{itemize}
\item We provide non-trivial deterministic and probabilistic upper bounds for a broad array of set-recognition problems. These bounds show that both deterministic and probabilistic models are non-trivial. Additionally, we prove a tight deterministic lower bound for equality. These initial bounds also provide a gentle introduction to this model of computation. We present these results in \cref{sec:upper-bounds}.

\item We extend and improve \cref{thm:grigoriev} above in two ways. In the statement above, we have $F_n(X,Y) = \IP_n(X,Y)$ and the lower bounds are for the sets $S(F_n) = \{(X,Y) \mid F_n(X,Y) \leq 0\}$ over $\R$ and $V(F_n) = \{(X,Y) \mid F_n(X,Y) = 0\}$ over $\C$.

As our first extension, we prove that $F_n$ can be more general and not necessarily restricted to $\IP_n$. In fact, we show that if  $F_n$  satisfies a set of conditions (we formally provide these below), then we can prove strong lower bounds. Thus, this provides a new framework for lower bounds; to prove a lower bound, one can simply check these conditions. These results are presented in \cref{sec:framework}.

We also provide tight lower bounds for several set-recognition problems in deterministic and probabilistic models. Along the way, we show that the lower bound in \cref{thm:grigoriev} can be improved to $2n-4$ (see \cref{thm:inner-product-set-recognition}). %
See \cref{sec:new-lower-bounds}.

\item Finally, we give two applications of our results. First, we show that the lower bounds above can be used to prove lower bounds for a certain class of algebraic algorithms. These algorithms scan the input from left to right. We call them algebraic scanners. They are reminiscent of streaming algorithms, but with crucial differences. These results appear in \cref{sec:scanner}.

Second, we show that our lower bound proofs work for a more general setting. Specifically, we show that our lower bound arguments extend beyond polynomials to a more general algebraic computational model, inspired by the Blum--Shub--Smale model~\cite{BCSS98}, for Alice, Bob, and the referee. These results appear in \cref{sec:bss}.
\end{itemize}

\subsubsection{Communication bounds and reductions}
\label{sec:intro-bounds}

We study the deterministic and probabilistic communication complexity of several natural sets. We prove reductions among these problems, allowing lower bounds for one set to be transferred to others.
For exact equality  some field $\F$, we write
\[
V(X-Y)
\coloneqq \{(X,X)\mid X\in\F^n\}
=V(X_1-Y_1,\ldots,X_n-Y_n).
\]
Thus, $V(X-Y)=V(\EQ_n)$ over $\R$, whereas $V(X-Y)\subsetneq V(\EQ_n)$ over $\C$ when $n\geq2$.

\begin{table}[H]
\centering
\renewcommand{\arraystretch}{1.35}
\begin{tabularx}{\textwidth}{@{}>{\centering\arraybackslash}p{2.4cm}>{\raggedright\arraybackslash}X@{}}
\toprule
\textbf{Problem} & \textbf{Description} \\
\midrule
$\EQ_n$ & The inputs are equal: $X=Y$. \\
\addlinespace
$\UEQ_n$ & The inputs agree up to a permutation of their coordinates: there exists $\sigma\in S_n$ such that $X_i=Y_{\sigma(i)}$ for every $i\in[n]$. \\
\addlinespace
$\SD_{n,k}$ & The inputs differ in at most $k$ coordinates: $\#\setdef{i\in[n]}{X_i\neq Y_i}\leq k$. \\
\addlinespace
$\IP_n$ & Over $\R$, recognize the threshold condition $\sum_{i=1}^nX_iY_i\leq0$; over $\C$, recognize the zero condition $\sum_{i=1}^nX_iY_i=0$. \\
\addlinespace
$\BF_{n,A}$ & For a fixed public matrix $A$, recognize $X^TAY\leq0$ over $\R$ where $A \in \R^{n \times n}$ and $X^TAY=0$ over $\C$, where $A \in \C^{n \times n}$. \\
\addlinespace
$\SI_n$ & The coordinate sets intersect: $X_i=Y_j$ for some $i,j\in[n]$, equivalently, $\prod_{i,j\in[n]}(X_i-Y_j)=0$. \\
\bottomrule
\end{tabularx}
\caption{Set-recognition problems considered in this paper. All problems are considered over $\F\in\{\R,\C\}$.}
\end{table}
Let $V(\UEQ_n)$ denote the set of $(X,Y)$ such that they agree up to a permutation of their coordinates. Similarly, let $V(\SD_{n,k})$ denote the set of inputs $(X,Y)$ such that they differ in at most $k$ coordinates.
\begin{theorem}
\label{thm:intro-upper-bound}
The probabilistic communication complexity of $V(\UEQ_n)$ is  $O(1)$  and that of $V(\SD_{n,k})$ is also $O(1)$ when $k = O(1)$.
\end{theorem}

Note that $V(\UEQ_n)$ is simply asking whether $X$ and $Y$ are equal as multisets. %
We show that checking this can be reduced to checking whether $n$ symmetric functions (polynomials) of $X$ and $Y$ are equal or not. Symmetric polynomials appear naturally in the argument because we are checking for a symmetric property. Thus, this observation allows us to obtain the same upper bound for $V(\UEQ_n)$ and $V(X-Y)$.
To prove the upper bound for $V(\SD_{n,k})$, there are several technical insights. The main argument hinges on creating $k+1$ vectors in $\F^{k+1}$ such that they are linearly dependent if $X =Y$, but they are linearly independent with high probability if $X$ and $Y$ differ at $\geq k+1$ coordinates.

Next, we consider bilinear forms, $\BF_{n,A} = X^T A Y$, where $A$ is an $n \times n$ matrix over $\F$ and $X,Y \in \F^n$. It is clear that when $A = I_n$, this is exactly the same as  $\IP_n$. In fact, if $A$ is a rank $r$ matrix then, for any $(X,Y)$, there exists $(\widetilde{X}, \widetilde{Y})$ such that $\BF_n(X,Y) = \IP_r(\widetilde{X}, \widetilde{Y})$. We observe that this conversion is local to Alice and Bob, i.e., given $X$ to Alice (and $Y$ to Bob), Alice can compute $\widetilde{X}$ herself (and Bob can compute $\widetilde{Y}$ himself). Thus, the bounds for set-recognition problems related to $\IP_r$ transfer to that of $\BF_{n,A}$, where rank of $A$ is $r$. This yields a family of problems whose complexity is neither necessarily as large as the full $2n$ bound nor as small as  $O(1)$.

Finally, we consider exact equality, $V(X-Y)$. A natural approach to a constant-cost equality protocol is a polynomial encoding $\phi:\F^n\to\F$: Alice and Bob send $\phi(X)$ and $\phi(Y)$, and the referee compares them.
For its correctness, such an encoding must be injective, which is impossible when $n>1$; see \cref{prop:no-dimension-reducing-injection}.
On restricted domains, however, injective polynomial encodings may exist. In particular, we construct one for $\Z^n$ in \Cref{obs:eq-integers}, yielding a constant-cost equality protocol for integral inputs. Such restrictions are generally unnatural for equality over $\R$ or $\C$.
Thus, equality on unrestricted inputs should, in general, be expected to be hard.
We prove a tight lower bound of $2n$ on the deterministic communication complexity of $V(X-Y)$.

\begin{theorem}[note={Deterministic communication complexity of equality},
  store=deterministicequalitytheorem,
  restate-keys={note={Restated}},
  label=thm:deterministic-equality]
  \label{thm:EQ-lbd}
For every $n\geq1$ and $\F\in\{\R,\C\}$, the deterministic communication complexity of $V(X-Y)$ is $2n$.
\end{theorem}
The theorem generalizes a lower bound from~\cite{BV2008}. Their proof uses ideas from projective geometry and requires that all the polynomials in the communication protocol are homogeneous.

We do not assume homogeneity. Our proof is an adaptation of the classical \emph{fooling-set argument} from communication complexity to the algebraic setting.

The idea can be described as follows. For a deterministic protocol $C$, and a root-to-leaf path $\pi$ in $C$, let $T_\pi$ denote the set of inputs that end up at the leaf of $\pi$. First, we observe that there must be an open set $U$ such that the diagonal inputs $(X,X)$ follows $\pi$ for every $X \in U$.
At a high level, this is a pigeonhole argument. Once we have such a leaf, then we can form a fooling set $\{(X,X) \mid X \in U \}$.

\subsubsection{Framework for probabilistic communication}
\label{sec:intro-framework}
To describe the results in this section, we need some notation.

\noindent \textbf{Notations.} Recall that for a polynomial $p(X,Y) \in \F[X,Y]$, $\HH_{X|Y}(p)$ is the mixed Hessian matrix. We will drop the suffix $X|Y$, if the partition of variables is clear from the context. Let $\nabla_Xp$ denote an $n \times 1$ vector in $\F[X,Y]^n$ such that the $i$th entry of the vector is $\partial p / \partial X_i$. Similarly define $\nabla_Yp$. We use $R_p$ to denote the determinant of the following $(n+1) \times (n+1)$ matrix.
\[R_p = \Det \begin{pmatrix}
    \HH_{X|Y}(p) & \nabla_Xp \\
    \left(\nabla_Yp\right)^T & 0
  \end{pmatrix}\]

We are now ready to state the new lower bound framework for probabilistic communication. To ease the exposition, we will focus on $\F = \R$ for the rest of the section.

\begin{theorem}[note={Tight lower bound for real communication},
  store=realthresholdtheorem,
  restate-keys={note={Restated}},
  label=thm:framework-R-threshold]
Let $F\in\R[X,Y]$ be irreducible, and let
$S(F)=\left\{(X,Y)\in\R^n\times\R^n\mid F(X,Y)\leq0\right\}.$
If there exists $u\in\R^n\times\R^n$ such that
\[
  F(u)=0
  \qquad\text{and}\qquad
  R_F(u)\neq0,
\]
then, the probabilistic communication complexity of $S(F)$ over $\R$ is $\geq 2n$, i.e., $\PCC_{\R}(S(F)) \geq 2n$.
\end{theorem}

\noindent \textbf{Proof technique.} The proof generalizes several technical ideas from~\cite{Gri2008}.

In fact, Grigoriev's theorem can be stated more generally than the statement of \cref{thm:grigoriev}, where the set one wants to recognize does not need to be $S(F)$ for a polynomial $F$, but instead can be any semialgebraic set $S$. Specifically, the general statement says that, for a semialgebraic set $S$, if $\dim(\partial S \cap V(\IP_n)) = 2n-1$, then its probabilistic communication complexity is at least $2n-6$.

Our first step is to observe that there is nothing special about $\IP_n$ in this setting. As long as the polynomial $F(X,Y)$ is irreducible and $F\nmid R_F$, where $R_F$ is the determinant polynomial defined above, then, for any semialgebraic set $S$ such that $\dim(\partial S \cap V(F)) = 2n-1$, we can obtain a lower bound for the set-recognition problem for $S$. This makes our framework widely applicable.

The second improvement comes from the fact we present a better analysis of the connection between the rank of $\HH_{X|Y}(F)$ and the communication complexity of communication protocols. As mentioned above, Abelson proves that if $\HH_{X|Y}(F)$ is high then the communication complexity of evaluating $F$ is also high. Grigoriev presents a technique to \emph{lift} lower bounds for polynomial evaluation problems to set-recognition problems for a related set. We fine-tune this lift to obtain better parameters and thus tight lower bounds.

The framework also has a complex analogue, with the boundary condition replaced with Zariski closure. This yields the following theorem.

\begin{theorem}[note={Tight lower bound for communication over $\C$},
  store=complexhypersurfacetheorem,
  restate-keys={note={Restated}},
  label=thm:framework-C-hypersurface]
Let $F\in\C[X,Y]$ be irreducible and let
$V(F)=\setdef{(X,Y)\in\C^n\times\C^n}{F(X,Y)=0}.$ If there exists $u\in\C^n\times\C^n$ such that
\[
  F(u)=0
  \qquad\text{and}\qquad
  R_F(u)\neq0,
\]
then $\PCC_{\C}(V(F)) \geq 2n$.
\end{theorem}

\subsubsection{New lower bounds}
\label{sec:new-lbds}
Our framework can be applied to obtain several new lower bounds.

Let $\EQ_n(X,Y)=\sum_{i=1}^n(X_i-Y_i)^2$ and $\EQ_{n,\varepsilon}=\EQ_n-\varepsilon$. For $\F=\R$, the polynomial $\EQ_n$ evaluates to $0$ if and only if $X=Y$. The set $S(\EQ_{n,\varepsilon})$ consists of pairs $(X,Y)$ whose squared Euclidean distance is at most $\varepsilon$. We show that its probabilistic communication complexity is $2n$. This contrasts the $O(1)$ upper bound for $V(\EQ_n)$. For $\F=\C$, we show that the probabilistic communication complexity of $V(\EQ_{n,\varepsilon})$ is also $2n$.

We also analyze $S(\IP_n)$ over $\R$ and $V(\IP_n)$ over $\C$, proving that their probabilistic communication complexities are at least $2n-4$. This slightly improves the lower bound of $2n-6$ obtained in~\cite{Gri2008}.

Next, we consider $V(\SI_n)$, where $\SI_n(X,Y) = \prod_{i \in [n]} \prod_{j \in [n]} (X_i - Y_j)$.
Note that $\SI_n(X,Y) = 0$ if and only if $X,Y$ have a non-empty intersection when thought of as sets. Grigoriev analyzed $V(\SI_n)$ and proved a probabilistic communication lower bound of $n/2$ over $\C$ and a lower bound of $n$ over $\R$.
We prove a lower bound of $2n-6$ for recognizing two closely related sets. Specifically, we show that the probabilistic communication complexity of $S(\SI_{n,\varepsilon})$ over $\R$ and $V(\SI_{n,\varepsilon})$ over $\C$ is at least $2(n-3)$, where $\SI_{n, \varepsilon} = \SI_n - \varepsilon$.

\subsubsection{Applications of the lower bounds and the framework}
\label{sec:intro-app}
We present two applications of the above communication lower bounds.
\paragraph*{Algebraic Scanners.}
We define a class of algebraic algorithms, which we call \emph{algebraic scanners}.
An algebraic scanner processes its input from left to right. At each step, it can compute an univariate polynomial evaluation of the current input and store the resulting value in a memory register. In addition, it can perform polynomial tests on the contents of its memory registers, with the outcomes of these tests determining its subsequent actions. After completing a full scan of the input, the scanner either accepts or rejects. An $(\ell,t)$ algebraic scanner stores at most $\ell$ elements in its memory and makes at most $t$ tests. The complexity of the scanner is $(\ell+t)$. We also consider probabilistic variants of algebraic scanners.

We show that communication lower bounds for set-recognition problems lead to lower bounds on the complexity of the algebraic scanners. This connection is closely inspired by the well-known relationship between communication complexity of Boolean functions and lower bounds for streaming algorithms~\cite{MuthukrishnanSurvey,AlonMatiasSzegedy}.

We prove lower bounds for the following two natural geometric problems: (a) Given $X, Y \in \F^n$, check whether the $\ell_2$ distance between $X$ and $Y$ is at most $\varepsilon$, for some fixed $\varepsilon$. (b) Given $n$ green points followed by $n$ red points, check if there exists a bichromatic pair that collides, i.e. there are two points of different color with the same value.
\paragraph*{Alice, Bob, the referee, and BSS machines.} So far, we have assumed that Alice and Bob communicate evaluations of polynomials and the referee performs polynomial tests. Thus, implicitly, the internal computations of all three parties can be carried out by algebraic circuits, possibly of exponential size. We ask whether our framework continues to apply when Alice, Bob, and the referee are allowed to use a more general computational model internally. In particular, we allow the parties to compute rational functions, and, additionally, allow Alice and Bob to perform local tests involving rational functions of their respective inputs. These extensions bring the internal computational model closer to the Blum--Shub--Smale (BSS) model.
We prove that our lower bound framework continues to apply even in this more general setting. The proofs require a careful adaptation of the techniques we developed in \cref{sec:hessian} and \cref{sec:framework}. One interesting observation is that \emph{lifting} lower bounds from polynomial evaluation problems to set-recognition problems continues to hold even when rational functions are allowed in place of polynomials.

\subsection{Related work}

As mentioned above, Abelson~\cite{Abelson80} first studied the communication complexity of the polynomial evaluation problem. This was further developed by Luo and Tsitsiklis~\cite{LT93}
in the setting of distributed algebraic computation. These papers study the amount of information that must be exchanged by parties who jointly compute algebraic functions whose inputs are distributed between them.

Grigoriev~\cite{Gri2008}, as well as \Blaser\ and Vicari~\cite{BV2008}, studied the set-recognition problem. While our work mainly builds on Grigoriev’s model, \Blaser\ and Vicari also give a compelling formulation of algebraic communication. In their model, the parties may exchange arbitrary real or complex numbers. In particular, messages are not necessarily restricted to evaluations of polynomials. Moreover, Alice and Bob communicate with each other directly, rather than with a referee.  Consequently, their later messages may depend not only on their own inputs, but also on the previous transcripts, including messages received from the other party. Some variants of their model allow comparisons and zero-tests. The final decision may be made by Alice or Bob and may depend on parts of the input that are never communicated with the other party throughout the protocol. A main distinction from our setting (which is also Grigoriev's setting) is that \Blaser\ and Vicari~\cite{BV2008} focus on deterministic communication complexity, whereas we also study probabilistic protocols.

Prior to~\cite{BV2008}, Kraj{\'{\i}}cek had also introduced communication games where two parties exchange arbitrary messages rather than polynomial evaluations. This work aimed at generalizing the well-studied Karchmer-Wigderson games~\cite{KW90} and at connecting proof complexity to real communication complexity and monotone real formulas and circuits \cite{Krajicek98}.

Monotone real circuits form another related line of work connecting real algebraic computation, communication, and proof complexity. Pudlák introduced monotone real circuits in the context of proof complexity, and later work of Hrubeš and Pudlák studied their relationship to monotone Boolean complexity and real communication protocols \cite{Pudlak97,HrubesPudlak17}. More recently, Applebaum, Beimel, Nir, Peter, and Pitassi~\cite{ABNPP}
connected monotone real circuits and separable variants of them to secret-sharing schemes.

Finally, we mention Ben-Or's algebraic computation tree lower bound~\cite{BenOr1983}, another foundational result concerning exact decision problems over the reals.
Our work differs from Ben-Or's in both the underlying model of computation and in the techniques used to prove lower bounds. Ben-Or's proofs rely on combinatorial properties of the set being recognized. Such arguments do not directly capture the difficulty in our model, where the main challenge is to bound the amount of algebraic interaction between the two sides of the input.

\subsection{Discussion and open problems}

We end the introduction with a discussion and several open problems.

\begin{itemize}
\item In our work, we have been able to show many tight lower bounds in deterministic as well as probabilistic communication models. Two problems which both Grigoriev's work and our work study are $S(\IP_n)$ over $\R$ and $V(\IP_n)$ over $\C$. Can we prove a lower bound of $2n$ for these problems? We are not able to apply our general framework for the polynomial $\IP_n$ because when $F = \IP_n$, there is no $u$ such that $F(u)=0$ and $R_F(u) \neq 0$. We prove the lower bound $2n-4$ by refining individual steps of the proof and analyzing them for these two sets. Could we instead use some other properties of $\IP_n$ to obtain a tight lower bound?
\item Our lower bound framework is quite easy to state and appears to be broadly applicable. It would be interesting to identify further natural set-recognition problems to which the framework applies, and to prove lower bounds for them.
\item As mentioned above, the communication model studied in~\cite{BV2008} is different from the model we study here. An interesting direction is to develop lower bound techniques for probabilistic communication in their model.
\end{itemize}

\section{Preliminaries}
\label{sec:prelims}

For $n \in \N$ and $n \geq 1$, let $X = \{X_1, \dots, X_{n}\}$ and $Y = \{Y_1, \dots, Y_{n}\}$ be two sets of variables throughout this paper. We denote by $S_n$ the symmetric group on $[n]$, and by $\F[X, Y]$ the ring of polynomials in the variables $X$ and $Y$ with coefficients over $\F \in \{\R,\C\}$.

For polynomials $p_1,\ldots,p_m\in\F[X,Y]$, their \emph{common zero set} (or variety) is defined by
\[
    V(p_1,\ldots,p_m)
    \coloneqq
    \setdef{(X,Y) \in \F^n\times\F^n}{p_j(X,Y)=0\text{ for every }j\in[m]}.
\]
When $m=1$, we simply write $V(p)$.
Likewise, the \emph{nonzero set} (or distinguished set) of $p$ is defined as the complement of $V(p)$:
$$
    D(p) \coloneqq \{(X,Y) \in \F^n\times\F^n \mid p(X,Y)\neq0\}.
$$
Finally, for $p \in \R[X,Y]$ the threshold set (also known as the sublevel set) of $p$ is defined as
$$
    S(p) \coloneqq \setdef{(X,Y) \in \R^n \times \R^n}{p(X,Y) \leq 0}.
$$

\begin{definition}[Constructible set]
A set $W\subseteq\C^N$ is \emph{constructible} if it is a finite union of
sets of the form
\[
  \left\{z\in\C^N:
  f_1(z)=\cdots=f_k(z)=0,\quad g(z)\neq0\right\},
\]
where $f_1,\ldots,f_k,g\in\C[z_1,\ldots,z_N]$.
\end{definition}

\begin{definition}[Boundary of a set]
  Let $S \subseteq \mathbb{R}^N$. The \emph{Euclidean boundary} of $S$, denoted $\partial S$, is defined as the set of all points $x \in \mathbb{R}^N$ such that every open ball centered at $x$ intersects both $S$ and its complement $\mathbb{R}^N \setminus S$. Formally,
\[
  \partial S = \{ x \in \mathbb{R}^N \mid \forall r > 0, \, B_r(x) \cap S \neq \emptyset \text{ and } B_r(x) \cap (\mathbb{R}^N \setminus S) \neq \emptyset \},
\]
where $B_r(x)$ denotes the open Euclidean ball of radius $r$ centered at $x$.
\end{definition}

The following proposition records the dimension of a hypersurface over $\R$ and $\C$.
The real case follows from the Implicit Function Theorem and the dimension theory of real algebraic sets; see \cite[Proposition~3.3.10, 3.3.11]{bochnak1998} and \cite[Theorem~5-1]{Spivak1965}.
The complex case follows from the Krull Principal Ideal Theorem.

\begin{proposition}
\label{prop:real_hypersurface}
Let $\F\in\{\R,\C\}$, and let $p\in\F[x_1,\ldots,x_N]$ be a non-constant irreducible polynomial.
If $\F=\C$, then
\[
  \dim V(p)=N-1.
\]
If $\F=\R$ and if there exists $x_0\in V(p)$ such that
$\nabla p(x_0)\ne\mathbf 0$, then
\[
  \dim V(p)=N-1.
\]
If, moreover, $\nabla p(x)\ne\mathbf 0$ for every
$x\in V(p)$, then $V(p)$ is a smooth real algebraic hypersurface.
\end{proposition}

One might hope to make communication highly efficient by uniquely encoding an input using fewer polynomial coordinates.
The following proposition rules out such a compression on any nonempty open set.

\begin{proposition}[\protect{\cite[Corollary 1.6.3]{Tao2014}}]
\label{prop:no-dimension-reducing-injection}
Let $\F\in\{\R,\C\}$, let $U\subseteq\F^n$ be a nonempty Euclidean open set, and let $0\leq r<n$.
No polynomial map $\Phi:U\to\F^r$ is injective.
\end{proposition}

In this work, we will also use the following classical bound on the number of zeros of a nonzero multivariate polynomial over a grid~\cite{Ore1922,DeMilloLipton1978,Zippel1979,Schwartz1980,LidlN1996}.

\begin{lemma}[Polynomial identity testing lemma]
\label{lem:pit}
Let $0\ne f\in\F[z_1,\ldots,z_m]$ have total degree at most $d$, and
let $A\subseteq\F$ be finite. If $\alpha_1,\ldots,\alpha_m$ are chosen
independently and uniformly from $A$, then
\[
  \Pr\left[f(\alpha_1,\ldots,\alpha_m)=0\right]
  \leq \frac{d}{|A|}.
\]
\end{lemma}

\subsection{Communication complexity of a \pep}

Let $\F\in\{\R,\C\}$.
Alice holds $X=(X_1,\ldots,X_n)$ and Bob holds $Y=(Y_1,\ldots,Y_n)$.
To evaluate $g$ at $(X,Y)$, Alice communicates the values of finitely many polynomials in $X$, and Bob communicates the values of finitely many polynomials in $Y$.
The value of $g(X,Y)$ is then computed by applying a polynomial $P$ to all the communicated values.
The cost is the total number of values communicated.

\begin{definition}[Communication Complexity of a Polynomial]
\label{def:comm-complexity}
    For $r_1,r_2\in\mathbb{N}$, a \emph{polynomial-evaluation representation} of $g\in\F[X,Y]$ consists of polynomials
    $a_1, \dots, a_{r_1} \in \F[X]$,
    $b_1, \dots, b_{r_2} \in \F[Y]$, and
    a polynomial $P$ over $\F$ in $r_1+r_2$ variables such that
    \[
        g(X, Y) \quad = \quad P\Big(a_1(X), \dots, a_{r_1}(X), b_1(Y), \dots, b_{r_2}(Y)\Big).
    \]
    The cost of this representation is $r_1+r_2$. The \emph{communication complexity} of $g$, denoted $c(g)$, is the minimum cost of any such representation. In particular, constant polynomials have communication complexity $0$.
\end{definition}

\begin{example}[Inner product]
  \label{ex:inner-product}
Consider the inner product polynomial $\IP_n(X,Y)=\sum_{i=1}^n X_iY_i$.
Alice communicates the $n$ values $a_i(X) = X_i$, and Bob communicates the $n$ values $b_i(Y) = Y_i$.
Taking
\[
  P(z_1,\ldots,z_{2n}) \quad =\quad \sum_{i=1}^n z_i z_{n+i},
\]
we obtain $\IP_n(X,Y) = P\left(a_1(X),\ldots,a_n(X),b_1(Y),\ldots,b_n(Y)\right)$.
Thus $c(\IP_n)\le 2n$.
\end{example}
In fact, for any polynomial $g \in \F[X,Y]$, $c(g)$ is bounded by $2n$.
In the upcoming sections, we will see how Grigoriev uses linear algebra to prove a matching lower bound.

\subsection{Protocol for set-recognition}\label{subsec:porto_Set_rec}

While the previous section dealt with evaluating a polynomial, we now turn
to the communication complexity of \emph{recognizing a set}. Let
$\F\in\{\R,\C\}$. Alice receives $X\in\F^n$, and Bob receives
$Y\in\F^n$. Let $S \subseteq \F^n \times \F^n$ be the set that they want to recognize, i.e., they want to check whether their input belongs to this set or not.

\begin{definition}[Deterministic set-recognition protocol]
\label{def:set-recognition-protocol}
A deterministic communication protocol over $\F$ is a finite rooted tree, with the root at level $1$, whose leaves are labeled Accept or Reject.
Each internal node $v$ at level $r$ is assigned
\begin{itemize}
    \item a message polynomial $q_v$, associated with one of the parties:
    either $q_v\in\F[X]$ or $q_v\in\F[Y]$; and
    \item a testing polynomial $P_v\in\F[z_1,\ldots,z_r]$.
\end{itemize}

On input $(X,Y)$, for a $v$ at level $r$, the associated party broadcasts $h_r = q_v(X)$ if the node $v$ belongs to Alice or $h_r = q_v(Y)$ if $v$ belongs to Bob.  The outgoing edge is determined by evaluating $P_v$ on the transcript $h=(h_1,\ldots,h_r)$, where $h_i$ are the message polynomials along a path from the root to $v$.
Over $\R$, the three possible outcomes are $P_v(h)<0$, $P_v(h)=0$, and $P_v(h)>0$. Over $\C$, the two possible outcomes are $P_v(h)=0$ and $P_v(h)\ne0$.
The outgoing edges are labeled by these outcomes, so every input determines a unique root-to-leaf
path.\footnote{Grigoriev~\cite{Gri2008} considers the more general model
in which each node may have a finite family of testing polynomials. All
results in this paper continue to hold in that model as well.}

The protocol \emph{recognizes} a set $S\subseteq\F^n\times\F^n$ if this path ends at an Accept leaf exactly when $(X,Y)\in S$.
The cost of the protocol is the maximum number of internal nodes on a root-to-leaf path.
The \emph{deterministic communication complexity} $\DCC_{\F}(S)$ is the minimum cost of a protocol recognizing $S$. We drop the suffix when the field is clear from the context.
\end{definition}

To familiarize ourselves with the model, consider the equality polynomial
\[
  \EQ_n(X,Y)=\sum_{i=1}^n(X_i-Y_i)^2
\]
and its zero set
\[
  V(\EQ_n)=\setdef{(X,Y)\in\F^n\times\F^n}{\EQ_n(X,Y)=0}.
\]
Over $\R$, $V(\EQ_n)=V(X-Y)$; over $\C$, $V(X-Y)\subsetneq V(\EQ_n)$ when $n\geq2$. Nevertheless, the following protocol recognizes $V(\EQ_n)$ over either field.

First, Alice broadcasts $h_i=X_i$ for $1\le i\le n$, and Bob then broadcasts $h_{n+i}=Y_i$ for $1\le i\le n$. The testing polynomials at the first $2n-1$ nodes are fixed nonzero constants, so the protocol simply continues. After the final message, the referee evaluates
\[
  P_{2n}(h)=\sum_{i=1}^n(h_i-h_{n+i})^2.
\]
The protocol accepts if this value is zero and rejects otherwise. Thus, it recognizes $V(\EQ_n)$ with cost $2n$ over both $\R$ and $\C$.

\begin{definition}[Probabilistic set-recognition protocol]\label{def:probabilistic-set-recognition-protocol}
A probabilistic communication protocol $\mathcal{C}$ consists of deterministic protocols
$C_1,\ldots,C_N$, for some finite $N$, as in \cref{def:set-recognition-protocol}, together with
probabilities $\mu_1,\ldots,\mu_N>0$ satisfying
$\sum_{i=1}^N\mu_i=1$. Here, $\mu$ is a probability distribution over the deterministic protocols $C_1, \ldots, C_N$.
The protocol recognizes
$S\subseteq\F^n\times\F^n$ with
bounded error if every input is classified correctly with probability at
least $2/3$; that is,
\[
  \Pr_{\mu}\left[\mathcal{C} \text{ accepts }(X,Y)\iff(X,Y)\in S\right]
  \geq \frac{2}{3}.
\]
Its cost is $\max_i\operatorname{cost}(C_i)$, and the \emph{probabilistic
communication complexity} $\PCC_{\F}(S)$ is the minimum cost of a
probabilistic protocol recognizing $S$ with bounded error. We use $\PCC(S)$ when the field is clear from the context.
\end{definition}

The following example demonstrates that probabilistic protocols can be significantly more efficient than deterministic protocols. (We will later show a tight lower bound for the same problem in the deterministic setting.)

\begin{example}[Equality]
The exact-equality set $V(X-Y)$ admits a probabilistic protocol of cost $2$.
\end{example}
\begin{proof}
Fix a set $A\subseteq\F$ of size $3n$, and choose $r\in A$ uniformly at
random. For each choice of $r$, Alice and Bob respectively broadcast
\[
  h_1=\sum_{i=1}^n X_i r^i
  \qquad\text{and}\qquad
  h_2=\sum_{i=1}^n Y_i r^i.
\]
The first testing polynomial is constant, while the second is
$P(z_1,z_2)=z_1-z_2$. The protocol accepts exactly when $P(h_1,h_2)=0$.
Thus every deterministic protocol in the distribution has cost $2$.

If $X=Y$, the protocol always accepts. If $X\ne Y$, then
\[
  Q(z)=\sum_{i=1}^n(X_i-Y_i)z^i
\]
is a nonzero polynomial of degree at most $n$. It has at most $n$ roots,
so the probability that $Q(r)=0$ is at most $n/(3n)=1/3$. Hence the
protocol rejects unequal inputs with probability at least $2/3$.
\end{proof}

\section{Communication bounds and reductions}
\label{sec:upper-bounds}

In this section, we prove communication bounds for natural sets such as equality and sparse difference. The discussion in this section will also serve as a warm-up to the communication models.

\subsection{Upper bounds for equality and its variant}

In \cref{subsec:porto_Set_rec}, we discussed a deterministic protocol for the set $V(\EQ_n)$.
We now turn to exact equality, $V(X-Y)$, and show that integral inputs admit a constant-cost deterministic protocol.
The following was also observed as a remark in \cite{Vicari2008}.

\begin{observation}[Equality over integers]
\label{obs:eq-integers}
Let $\F\in\{\R,\C\}$ and $n \geq 1$. Under the promise that $X,Y\in\Z^n$, exact equality admits a deterministic protocol over $\F$ of complexity exactly $2$.
\end{observation}

\begin{proof}
  We will first construct an injective polynomial map for Alice and Bob, and then use it to construct a depth-$2$ protocol for recognizing equality over the integers.

  \textbf{Injective polynomial map.} We first construct an injective polynomial
  $\Phi_n:\Z^n\to\Z_{\geq0}$. Define
  \[
    \varphi(z)=2z^2-z
    \qquad\text{and}\qquad
    \psi(u,v)=\left(u+v\right)^2+v.
  \]
  The map $\varphi:\Z\to\Z_{\geq0}$ is injective: if
  $\varphi(z)=\varphi(w)$, then
  \[
    (z-w)\left(2(z+w)-1\right)=0,
  \]
  and the second factor cannot vanish for integers $z,w$. The map
  $\psi:\Z_{\geq0}^2\to\Z_{\geq0}$ is also injective. Suppose
  $\psi(u,v)=\psi(u',v')$, and set $s=u+v$ and $s'=u'+v'$. If $s>s'$,
  then $s\geq s'+1$, and therefore
  \[
    \psi(u,v)\geq s^2\geq(s'+1)^2
    >s'^2+s'\geq\psi(u',v').
  \]
  Since $\psi(u,v) > \psi(u',v')$, we have a contradiction. By symmetry, $s<s'$ is also impossible, so $s=s'$.
  It follows that $s^2+v=s^2+v'$, hence $v=v'$ and then $u=u'$.

  Now define recursively
  \[
    \Phi_1(z_1)=\varphi(z_1),
    \qquad
    \Phi_k(z_1,\ldots,z_k)
    =\psi\left(\Phi_{k-1}(z_1,\ldots,z_{k-1}),\varphi(z_k)\right).
  \]
  Since $\varphi$ and $\psi$ are injective polynomial maps, so is
  $\Phi_n$. These maps have integer coefficients and therefore define valid
  message polynomials over both $\R$ and $\C$.

  \textbf{Set-recognition protocol.}
  Consider the following depth-$2$ protocol. At the root, Alice uses the
  message polynomial $q_1(X)=\Phi_n(X)$ and broadcasts
  $h_1=q_1(X)$. The testing polynomial at this node is the constant
  polynomial $1$, so the protocol continues to the second node. There,
  Bob uses the message polynomial $q_2(Y)=\Phi_n(Y)$ and broadcasts
  $h_2=q_2(Y)$. The testing polynomial is
  \[
    P(z_1,z_2)=z_1-z_2.
  \]
  The protocol accepts on the zero branch and rejects otherwise. Hence, on
  the promised inputs, it accepts precisely when
  $\Phi_n(X)=\Phi_n(Y)$, which, by injectivity, is equivalent to $X=Y$.

  Finally, a depth-$1$ protocol has only one message, belonging to one
  party, and its output therefore depends only on that party's input. It
  cannot recognize equality even on the promised inputs, since the answer
  depends on both inputs. Thus the complexity is exactly $2$.
\end{proof}

Next, we consider unordered equality, in which Alice's and Bob's inputs are viewed as multisets rather than ordered vectors.
Equivalently, their inputs are equal up to a permutation of the coordinates.

\begin{observation}[Unordered equality reduces to equality]
\label{obs:unordered-equality}
Let $\F$ be any field, and define
\begin{equation}
  V(\UEQ_n)
  =\setdef{(X,Y)\in\F^n\times\F^n}
  {\text{there exists }\sigma\in S_n\text{ such that }
    X_i=Y_{\sigma(i)}\text{ for all }i\in[n]}.
    \label{eq:ueq-def}
\end{equation}
Every deterministic or probabilistic protocol recognizing $V(X-Y)$ gives a protocol of the same cost recognizing $V(\UEQ_n)$. Consequently,
\[
  \DCC_{\F}\left(V(\UEQ_n)\right)
  \leq \DCC_{\F}\left(V(X-Y)\right)
  \qquad\text{and}\qquad
  \PCC_{\F}\left(V(\UEQ_n)\right)
  \leq \PCC_{\F}\left(V(X-Y)\right).
\]
\end{observation}

\begin{proof}
For $Z\in\F^n$, let $\ESym_k(Z)$ denote the $k$th elementary symmetric polynomial evaluated at $Z$, and define
\[
  \ESym(Z)=\left(\ESym_1(Z),\ldots,\ESym_n(Z)\right).
\]
The vector $\ESym(Z)$ determines the monic polynomial with roots
$Z_1,\ldots,Z_n$:
\begin{equation}
  \prod_{i=1}^n(t-Z_i)
  \quad =\quad t^n-\ESym_1(Z)t^{n-1}+\cdots+(-1)^n\ESym_n(Z). \label{eq:esym-generator}
\end{equation}
Consequently, the coordinates of $X$ and $Y$ agree as multisets if and only if $\ESym(X)=\ESym(Y)$. Thus
\[
  (X,Y)\longmapsto \left(\ESym(X),\ESym(Y)\right)
\]
transforms an instance of unordered equality into an instance of equality, with
\[
  (X,Y)\in V(\UEQ_n)
  \quad\Longleftrightarrow\quad
  \left(\ESym(X),\ESym(Y)\right)\in V(X-Y).
\]
Indeed, the forward implication follows because every elementary symmetric polynomial is invariant under permutations.
Conversely, if $\ESym(X)=\ESym(Y)$, then comparing coefficients in \cref{eq:esym-generator} gives
\[
  \prod_{i=1}^n(t-X_i)=\prod_{i=1}^n(t-Y_i).
\]
By unique factorization in $\F[t]$, these polynomials have the same roots with the same multiplicities.
Thus the coordinates of $X$ and $Y$ agree as multisets, so $(X,Y)\in V(\UEQ_n)$.

The reduction is local: Alice computes $\ESym(X)$ and Bob computes $\ESym(Y)$.
Given a protocol for $V(X-Y)$, replace each Alice-message polynomial $q(X)$ by $q(\ESym(X))$ and each Bob-message polynomial $q(Y)$ by $q(\ESym(Y))$.
These remain one-party polynomials, so the resulting protocol recognizes $V(\UEQ_n)$.
\end{proof}

\subsection{Deterministic lower bound for equality}

We now prove that deterministic recognition of $V(X-Y)$ requires each party to communicate $n$ polynomial values. This set captures exact equality over both $\R$ and $\C$.
The argument begins by finding one accepting path followed by an open set of diagonal inputs which traverse that path.

\begin{lemma}[A common accepting path]
\label{lem:equality-common-path}
Let $\F\in\{\R,\C\}$, let $U\subseteq\F^n$ be a nonempty Euclidean open set, and let $C$ be a deterministic protocol recognizing $V(X-Y)$ under the promise $X,Y\in U$.
There exist a nonempty Euclidean open set $W\subseteq U$ and an accepting root-to-leaf path $\pi$ in $C$ such that $(X,X)$ follows $\pi$ for every $X\in W$.
\end{lemma}

\begin{proof}
For every internal node $v$, let
\[
  A_v(X)=\left(a_{v,1}(X),\ldots,a_{v,s_v}(X)\right)
  \quad\text{and}\quad
  B_v(Y)=\left(b_{v,1}(Y),\ldots,b_{v,t_v}(Y)\right)
\]
collect Alice's and Bob's message polynomials along the path from the root to $v$, including the message at $v$. Let
\[
  \left(h_{v,1}(X,Y),\ldots,h_{v,r_v}(X,Y)\right),
  \qquad r_v=s_v+t_v,
\]
be their chronological interleaving. Thus, each $h_{v,j}$ is either some $a_{v,k}(X)$ or some $b_{v,\ell}(Y)$. Define the restriction of the testing polynomial at $v$ to the diagonal by
\[
  \widetilde P_v(X)
  \coloneqq P_v\left(h_{v,1}(X,X),\ldots,h_{v,r_v}(X,X)\right)\in\F[X].
\]
Let
\[
  P(X)=\prod_{\substack{v\text{ internal}\\ \widetilde P_v\not\equiv0}}
  \widetilde P_v(X),
\]
where the empty product is $1$.
The tree is finite, hence $P$ is well-defined. Moreover, by definition, $P$ is a nonzero polynomial. %
Choose $X_0\in U$ such that $P(X_0)\ne0$; such a point exists because a nonzero polynomial cannot vanish on a nonempty Euclidean open set.

Let $\pi=(v_1,\ldots,v_\ell)$ be the path followed by $(X_0,X_0)$.
For each $i\in[\ell]$ with $\widetilde P_{v_i}\not\equiv0$, continuity gives an open ball $B_i$ containing $X_0$ on which $\widetilde P_{v_i}$ has the same sign as at $X_0$ over $\R$, or remains nonzero over $\C$.
If $\widetilde P_{v_i}\equiv0$, set $B_i=\F^n$.
Then
\[
  W=U\cap\bigcap_{i=1}^{\ell}B_i
\]
is a nonempty Euclidean open subset of $U$, since $\pi$ is finite and $U$ and every $B_i$ are open sets containing $X_0$.
Every $(X,X)$, with $X\in W$, takes the same branch as $(X_0,X_0)$ at each node and therefore it follows $\pi$.
Finally, $\pi$ is accepting because $(X_0,X_0)\in V(X-Y)$.
\end{proof}

Let $a_1,\ldots,a_s\in\F[X]$ and $b_1,\ldots,b_t\in\F[Y]$ be the transcript messages  along $\pi$.
Define
\[
  A_\pi(X)=\left(a_1(X),\ldots,a_s(X)\right),
  \qquad
  B_\pi(X)=\left(b_1(X),\ldots,b_t(X)\right).
\]
Next we show that the common path forces each of these maps to distinguish all inputs in $W$.

\begin{lemma}
\label{lem:equality-message-injectivity}
Let $W$ and $\pi$ be given by \cref{lem:equality-common-path}.
The maps $A_\pi:W\to\F^s$ and $B_\pi:W\to\F^t$ are injective.
\end{lemma}

\begin{proof}
Suppose that $X,X'\in W$ and $A_\pi(X)=A_\pi(X')$.
Along $\pi$, the cross-input $(X,X')$ produces the same transcript as $(X',X')$: Alice's messages agree by assumption, and Bob has input $X'$ in both instances.
Thus $(X,X')$ reaches the same accepting leaf.
Since $(X,X')\in U\times U$ and $C$ recognizes equality on this domain, $X=X'$, and hence $A_\pi$ is injective on $W$.

If $B_\pi(X)=B_\pi(X')$, then $(X,X')$ produces the same transcript as $(X,X)$ and is likewise accepted.
Therefore $X=X'$, so $B_\pi$ is also injective on $W$.
\end{proof}

We now have all the ingredients for the lower bound.

\begin{theorem}[note={Strengthened restatement of \cref{thm:deterministic-equality}}]
\label{thm:deterministic-equality-open}
For every $n\geq1$, $\F\in\{\R,\C\}$, and nonempty Euclidean open set $U\subseteq\F^n$, the deterministic communication complexity of recognizing $V(X-Y)$ under the promise $X,Y\in U$ is $2n$.
Consequently, the deterministic communication complexity of $V(X-Y)$ is $2n$.
\end{theorem}

\begin{proof}
Broadcasting all coordinates of $X$ and $Y$ gives a deterministic protocol of cost $2n$.

Conversely, let $U\subseteq\F^n$ be a nonempty Euclidean open set, and let $C$ be any deterministic protocol deciding equality under the promise $X,Y\in U$.
By \cref{lem:equality-common-path}, the diagonal inputs from some nonempty open set $W\subseteq U$ follow one accepting path $\pi$.
Let $s$ and $t$ be the numbers of Alice- and Bob-messages along $\pi$.
By \cref{lem:equality-message-injectivity}, the corresponding polynomial maps $A_\pi:W\to\F^s$ and $B_\pi:W\to\F^t$ are injective.
By \cref{prop:no-dimension-reducing-injection}, this gives $s\geq n$ and $t\geq n$.
Hence $\pi$ contains at least $s+t\geq2n$ messages, so the depth of $C$ is at least $2n$.
Taking $U=\F^n$ gives the claimed complexity of $V(X-Y)$.
\end{proof}

By restricting the inputs, the tight lower bound for $V(X-Y)$ can be lifted to $V(\UEQ_n)$ defined in \cref{eq:ueq-def}.

\begin{theorem}
\label{thm:deterministic-unordered-equality}
For every $n\geq1$ and $\F\in\{\R,\C\}$,
\[
  \DCC_{\F}\left(V(\UEQ_n)\right)=2n.
\]
\end{theorem}

\begin{proof}
The upper bound follows from \cref{obs:unordered-equality} and \cref{thm:deterministic-equality}.

For the lower bound, choose pairwise disjoint nonempty open sets
$U_1,\ldots,U_n\subseteq\F$ and set $U=U_1\times\cdots\times U_n$.
For $X,Y\in U$, if $X_i=Y_{\sigma(i)}$, then, they belong to both $U_i$ and $U_{\sigma(i)}$, disjointness forces $\sigma(i)=i$.
Thus
\[
  \text{For all }X,Y\in U \qquad (X,Y)\in V(\UEQ_n)
  \quad\Longleftrightarrow\quad
  X=Y.
\]
Thus any protocol recognizing $V(\UEQ_n)$ decides equality under the promise $X,Y\in U$. Since $U$ is nonempty and Euclidean open, \cref{thm:deterministic-equality-open} gives $\DCC(V(\UEQ_n))\geq2n$.
\end{proof}

\subsection{Complexity of sparse difference}
\label{sec:com-SD}

Let $\F\in\{\R,\C\}$.
For integers $n\geq 1$ and $1\leq k\leq n$, define the \emph{$k$-sparse difference set over $\F$} by:
\[
  V(\SD_{n,k})
  \coloneqq
  \setdef{(X,Y)\in\F^n\times\F^n}
  {\#\{i\in[n]:X_i\ne Y_i\}\leq k}.
\]
We start with the observation that any set-recognition protocol for $V(\SD_{n+k,k})$ would give an equally efficient protocol for $V(X-Y)$.

\begin{observation}[Equality reduces to sparse difference]
\label{thm:equality-to-sparse-difference}
For every $n\geq1$ and $k\geq1$,
\[
  \DCC_{\F}\left(V(X-Y)\right)
  \leq \DCC_{\F}\left(V(\SD_{n+k,k})\right)
  \qquad\text{and}\qquad
  \PCC_{\F}\left(V(X-Y)\right)
  \leq \PCC_{\F}\left(V(\SD_{n+k,k})\right).
\]
\end{observation}

\begin{proof}
Given $(X,Y)\in\F^n\times\F^n$, define
\[
  \widehat{X}=\left(X_1,\ldots,X_n,\underbrace{0,\ldots,0}_{k}\right)
  \qquad\text{and}\qquad
  \widehat{Y}=\left(Y_1,\ldots,Y_n,\underbrace{1,\ldots,1}_{k}\right).
\]
The appended coordinates contribute exactly $k$ disagreements, and hence
\[
  \left|\setdef{i\in[n+k]}{\widehat{X}_i\ne\widehat{Y}_i}\right|
  =k+\left|\setdef{i\in[n]}{X_i\ne Y_i}\right|.
\]
Therefore,
\[
  X=Y
  \quad\Longleftrightarrow\quad
  \left(\widehat{X},\widehat{Y}\right)\in V(\SD_{n+k,k}).
\]
\end{proof}

Together with the deterministic lower bound for equality, this reduction gives the following lower bound for sparse difference:

\begin{theorem}[Deterministic lower bound for $V(\SD_{n,k})$]
\label{thm:deterministic-sparse-difference}
For every $n\geq1$, $1\leq k\leq n$, and $\F\in\{\R,\C\}$,
\[
  \DCC_{\F}\left(V(\SD_{n,k})\right)\geq 2(n-k).
\]
\end{theorem}

\begin{proof}
For $k=n$, the claim is immediate.
For $k<n$, apply \Cref{thm:equality-to-sparse-difference} with equality on $\F^{n-k}$.
By \cref{thm:EQ-lbd},
\[
  \DCC\left(V(\SD_{n,k})\right)
  = \DCC\left(V(\SD_{(n-k)+k, ~k})\right)
  \geq \DCC\left(V(X-Y)\right) = 2(n-k).
\]
\end{proof}

In the next theorem, we show that $V(\SD_{n,k})$ admits an efficient probabilistic protocol.
For the correctness of the protocol, we will need the following standard evaluation-rank criterion:

\begin{proposition}[\protect{\cite[Claim~7]{Kayal2011}}]
\label{prop:kayal-evaluation-criterion}
Let $f_1,\ldots,f_\ell\in\F[t]$. Then, the matrix
\[
  \begin{pmatrix}
    f_1(\alpha_1) & \cdots & f_\ell(\alpha_1) \\
    \vdots & & \vdots \\
    f_1(\alpha_\ell) & \cdots & f_\ell(\alpha_\ell)
  \end{pmatrix}
\]
has nonzero determinant as a polynomial in
$\alpha_1,\ldots,\alpha_\ell$ if and only if
$f_1,\ldots,f_\ell$ are linearly independent over $\F$.
\end{proposition}

We now apply this criterion to obtain an efficient probabilistic protocol
for recognizing $V(\SD_{n,k})$.

\begin{theorem}[note={Probabilistic protocol for $V(\SD_{n,k})$},
  label=thm:probabilistic-sparse-difference]
Let $\F\in\{\R,\C\}$. For $n\geq1$ and $1\leq k\leq n$, let
\[
  V(\SD_{n,k})
  =\setdef{(X,Y)\in\F^n\times\F^n}
  {\#\{i\in[n]:X_i\neq Y_i\}\leq k}.
\]
Then,
\[
  \PCC_{\F}\left(V(\SD_{n,k})\right)\leq 2k+2.
\]
In particular, for fixed $k$, the probabilistic communication complexity is independent of $n$.
\end{theorem}

\begin{proof}
\textbf{Setup. }
For inputs $X,Y\in\F^n$, define
\[
  q_X(t)=\sum_{i=1}^n X_it^i,
  \qquad
  q_Y(t)=\sum_{i=1}^n Y_it^i,
  \qquad
  q_{\mathrm{diff}}(t)=q_X(t)-q_Y(t).
\]
Let
\[
  D=(k+1)n\binom{n}{k}.
\]

Choose a finite set $A\subseteq\F$ with $|A|\geq3D$.
Choose
\[
  \vecalpha=(\alpha_1,\ldots,\alpha_{k+1})\in A^{k+1}.
\]
uniformly at random.
For a set $S=\{s_1,\ldots,s_k\}\subseteq[n]$ of size $k$,

\[
  M_S(\vecalpha)=
  \begin{pmatrix}
    \alpha_1^{s_1} & \cdots & \alpha_1^{s_k} & q_{\mathrm{diff}}(\alpha_1) \\
    \vdots & & \vdots & \vdots \\
    \alpha_{k+1}^{s_1} & \cdots & \alpha_{k+1}^{s_k}
      & q_{\mathrm{diff}}(\alpha_{k+1})
  \end{pmatrix}.
\]
The sampled value of $\vecalpha$ determines the following deterministic
protocol $T_{\vecalpha}$.

\textbf{Protocol. }
Alice broadcasts $q_X(\alpha_j)$ and Bob broadcasts $q_Y(\alpha_j)$ for
every $j\in[k+1]$. All intermediate testing polynomials are constant. At
the final node, the protocol computes the transcript differences
\[
  q_{\mathrm{diff}}(\alpha_j)=q_X(\alpha_j)-q_Y(\alpha_j)
  \qquad (j\in[k+1]).
\]
Define the testing polynomial
\[
  P(\vecalpha)
  =
  \prod_{S\in\binom{[n]}{k}}\Det M_S(\vecalpha).
\]
This is a polynomial in the transcript values $q_X(\alpha_j),q_Y(\alpha_j)$, and the protocol accepts exactly when it evaluates to zero.
Thus each $T_{\vecalpha}$ has depth $2k+2$.

\textbf{Correctness.}
The input $(X,Y)\in V(\SD_{n,k})$ if and only if $q_{\mathrm{diff}}$ has at most $k$ nonzero
coefficients.
Then, for some $S\in\binom{[n]}{k}$ the support of $q_{\mathrm{diff}}$ is contained in $S$.
Indeed, if the support has fewer than $k$ elements, we may extend it to a $k$-element subset of $[n]$.
For this $S$, the last column of $M_S(\vecalpha)$ is a linear combination of the first $k$ columns. Hence, $P(\vecalpha)=0$ for every choice of $\vecalpha$, so the sampled protocol always accepts.

Conversely, suppose $(X,Y)\notin V(\SD_{n,k})$.
Then, $q_{\mathrm{diff}}$ has more than $k$ nonzero coefficients. For every $S=\{s_1,\ldots,s_k\}$, the polynomials
\[
  t^{s_1},\ldots,t^{s_k},q_{\mathrm{diff}}(t)
\]
are linearly independent, since $q_{\mathrm{diff}}$ has a nonzero coefficient in a
degree outside $S$.
By \cref{prop:kayal-evaluation-criterion}, each $\Det M_S(\vecalpha)$ is a nonzero polynomial in $\vecalpha$; hence, their product is nonzero.
Moreover,
\[
  \deg\Det M_S(\vecalpha)
  \leq n+\sum_{s\in S}s
  \leq(k+1)n,
\]
and therefore $\deg P\leq D$. By \cref{lem:pit},
\[
  \Pr
  \left[T_{\vecalpha}\text{ accepts }(X,Y)\right]
  =
  \Pr_{\vecalpha\in A^{k+1}}\left[P(\vecalpha)=0\right]
  \leq\frac{D}{|A|}\leq\frac13.
\]
Thus the protocol rejects with probability at least $2/3$. Alice and Bob
each communicate $k+1$ values, so the total cost is $2k+2$.
\end{proof}

\subsection{Complexity of bilinear forms}
\label{sec:bilinear-forms}
Let $A\in\F^{n \times n}$ be a fixed public matrix known to both Alice and Bob, and let $X=(X_1,\ldots,X_n)$ and $Y=(Y_1,\ldots,Y_n)$ be vectors of variables.
The matrix $A$ defines the bilinear form
\[
  \BF_{n,A}(X,Y)=X^TAY.
\]

Consider the real threshold set for a fixed $A\in\R^{n \times n}$
\[
  S(\BF_{n,A})
  =\setdef{(X,Y)\in\R^n\times\R^n}{\BF_{n,A}(X,Y)\leq0}.
\]
Over $\C$, we consider the zero set for a fixed $A\in\C^{n \times n}$
\[
  V(\BF_{n,A})
  =\setdef{(X,Y)\in\C^n\times\C^n}{\BF_{n,A}(X,Y)=0}.
\]

Recall the notation for the inner product problems from \cref{thm:grigoriev}.
Taking $A=I_n$ gives $\BF_{n,I_n}=\IP_n$.
Thus, the family of set-recognition problems for bilinear-form threshold sets includes the inner-product problem.
Moreover, we observe the following equivalence in the complexity of these two problems.

\begin{observation}[Equivalence with inner product]
\label{obs:bilinear-inner-product-equivalence}
Let $A\in\R^{n\times n}$ and $r=\rank_{\R}(A)\geq1$.
Then,
\[
  \DCC_\R\left(S(\BF_{n,A})\right)
  =\DCC_\R\left(S(\IP_r)\right),
  \qquad
  \PCC_\R\left(S(\BF_{n,A})\right)
  =\PCC_\R\left(S(\IP_r)\right).
\]
For the complex case, let $A\in\C^{n\times n}$ and $r=\rank_{\C}(A)\geq1$.
Then,
\[
  \DCC_\C\left(V(\BF_{n,A})\right)
  =\DCC_\C\left(V(\IP_r)\right),
  \qquad
  \PCC_\C\left(V(\BF_{n,A})\right)
  =\PCC_\C\left(V(\IP_r)\right).
\]
\end{observation}
\begin{proof}
First, let us assume that $\F = \R$. Choose $P,Q\in\operatorname{GL}_n(\R)$ such that
\[
  P^TAQ=
  \begin{pmatrix}
    I_r & \mathbf{0} \\
    \mathbf{0} & \mathbf{0}
  \end{pmatrix}.
\]
Let $\widetilde{X}\in\R^r$ and $\widetilde{Y}\in\R^r$ be the truncations of $P^{-1}X$ and $Q^{-1}Y$, respectively, to their first $r$ coordinates.
Then,
\begin{align*}
  \BF_{n,A}(X,Y)
  \quad &= \quad\left(P^{-1}X\right)^T P^TAQ \left(Q^{-1}Y\right)\\
  \quad &= \quad\left(P^{-1}X\right)^T
  \begin{pmatrix}
    I_r & \mathbf{0} \\
    \mathbf{0} & \mathbf{0}
  \end{pmatrix}
  \left(Q^{-1}Y\right)\\
  \quad &=\quad\widetilde{X}^{T}\widetilde{Y}
  \quad=\quad\IP_r\left(\widetilde{X},\widetilde{Y}\right).
\end{align*}
Consequently,
\[
  (X,Y)\in S(\BF_{n,A})
  \quad\Longleftrightarrow\quad
  (\widetilde{X},\widetilde{Y})\in S(\IP_r).
\]
By a similar argument for $\F= \C$ we also get the following:
\[
  (X,Y)\in V(\BF_{n,A})
  \quad\Longleftrightarrow\quad
  (\widetilde{X},\widetilde{Y})\in V(\IP_r).
\]
Here Alice computes $\widetilde{X}$ from $X$, while Bob computes $\widetilde{Y}$ from $Y$, so this transformation is local.

Conversely, let $(\widetilde{X}, \widetilde{Y})\in\F^r\times\F^r$ be an input to the inner product problem, where $\F\in\{\R,\C\}$.
Alice and Bob locally construct
\[
  X=P\left(\widetilde{X},\mathbf{0}_{n-r}\right),
  \qquad
  Y=Q\left(\widetilde{Y},\mathbf{0}_{n-r}\right).
\]
Then,
\[
  \BF_{n,A}(X,Y)
  =\left(\widetilde{X},\mathbf{0}_{n-r}\right)^TP^TAQ
    \left(\widetilde{Y},\mathbf{0}_{n-r}\right)
  =\widetilde{X}^T\widetilde{Y}
  =\IP_r\left(\widetilde{X},\widetilde{Y}\right).
\]
Thus, a protocol for the bilinear-form problem can recognize inner product with the same cost and error probability.
Together with the forward reduction, this proves the claimed equivalences.
\end{proof}

The observation transfers communication bounds for $r$-dimensional inner product directly to bilinear forms of rank $r$.
In particular, after proving the inner product lower bound in \cref{sec:inner-product}, we use it to obtain the same $2r-4$ lower bound for both the real threshold set $S(\BF_{n,A})$ and the complex zero set $V(\BF_{n,A})$.

\section{Communication complexity via mixed Hessian rank}
\label{sec:hessian}

We now return to the communication complexity of computing a polynomial, as defined in \cref{def:comm-complexity}; see \cref{ex:inner-product} for the basic example.
Abelson~\cite{Abelson80} proved that such a computation is constrained by the rank of the matrix of mixed second partial derivatives.

\begin{definition}[Mixed Hessian matrix]
\label{def:comm-matrix}
    Let $\F \in \{\R,\C\}$. For a polynomial $g \in \F[X,Y]$, its $n\times n$ \emph{mixed Hessian matrix} $\HH_{X|Y}(g)$ is defined by
    \[
        [\HH_{X|Y}(g)]_{ij} \quad = \quad
        \frac{\partial^2 g}{\partial X_i\partial Y_j}.
    \]
\end{definition}

The rank of the mixed Hessian matrix defined above is a lower bound on the communication complexity of computing a polynomial.
\begin{lemma}[note={\protect{\cite[Lemma 2.2]{Gri2008}}},
  store=commmatrixlemma,
  restate-keys={note={Restated}},
  label=lem:comm-matrix]
Let $\F\in\{\R,\C\}$, and let $g\in\F[X,Y]$. Suppose that
\[
  g(X,Y) = Q\bigl(a_1(X),\ldots,a_{r_1}(X), b_1(Y),\ldots,b_{r_2}(Y) \bigr)
\]
is a minimum-size representation of $g$, where $a_1,\ldots,a_{r_1}\in\F[X]$ and $b_1,\ldots,b_{r_2}\in\F[Y]$.
Then $\rank \HH_{X|Y}(g)\leq \min\{r_1,r_2\}$ and
\[
  c(g)=r_1+r_2\geq 2 \cdot \rank \HH_{X|Y}(g).
\]
Here the rank of a polynomial matrix is taken over $\F(X,Y)$.
\end{lemma}

For completeness, we give the proof of the lemma in \cref{app:comm-matrix-proof}.
The rank lemma can be used to obtain tight bounds for communication complexity of computing the inner product polynomial.

\begin{example}[Tightness for inner product]\label{ex:tight-innter-product}
Recall from \cref{ex:inner-product} that the inner product polynomial
\[
  \IP_n(X,Y)=\sum_{i=1}^n X_iY_i
\]
satisfies $c(\IP_n)\leq 2n$.
Clearly, $\HH_{X|Y}(\IP_n)$ is the identity matrix, hence, $\rank\HH_{X|Y}(\IP_n)=n$.
By \cref{lem:comm-matrix}, $c(\IP_n)\geq 2n$. Therefore, $c(\IP_n)=2n$.
\end{example}

Grigoriev observed that the rank of the mixed Hessian matrix is dictated by the irreducible factors of the polynomial. The following lemma generalizes \cite[Lemma 2.4]{Gri2008}.

\begin{lemma}[Rank of multiples]
\label{lem:general-rank}
Over a field $\F$ of characteristic zero\footnote{Over arbitrary fields $\F$, we require $\char(\F) \nmid m$}, let $p\in\F[X,Y]$ be irreducible, and let $r$ be an integer with $1\leq r\leq n$.
Suppose that some $r\times r$ minor of $\HH_{X|Y}(p)$ is
not divisible by $p$.
Let
\[
  g=p^m h,
  \qquad m\ge1,
  \qquad p\nmid h.
\]
Then,
\[
  \rank \HH_{X|Y}(g)\ge r-2.
\]
\end{lemma}

\begin{proof}
For each $i,j\in[n]$, we first differentiate with respect to $Y_j$ and then with respect to $X_i$, following the convention in \cref{def:comm-matrix}.
Differentiating $g=p^m h$ and grouping the rank-one terms gives
\begin{align*}
\HH_{X|Y}(g)
&=m p^{m-1}h\,\HH_{X|Y}(p)+p^m\HH_{X|Y}(h) \\
&\quad +m(\nabla_Xp)\left(\nabla_Y(p^{m-1}h)\right)^T
     +m p^{m-1}(\nabla_Xh)(\nabla_Yp)^T.
\end{align*}
After factoring out the common scalar $p^{m-1}$ from the first two terms, the main matrix is
\begin{align*}
  M &= mh\,\HH_{X|Y}(p)+p\,\HH_{X|Y}(h)\\
  &\equiv mh\,\HH_{X|Y}(p)\pmod p.
\end{align*}

Let $\mu$ be an $r\times r$ minor of $\HH_{X|Y}(p)$ that is not divisible by
$p$, and let $\mu_M$ be the corresponding minor of $M$. Then
\[
  \mu_M
  \equiv (mh)^r\mu
  \pmod p.
\]
Since $p$ is irreducible, it is prime in the unique factorization domain
$\F[X,Y]$. Also, $m$ is a nonzero scalar in $\F$ and hence a unit.
Together with the assumptions $p\nmid h$ and $p\nmid\mu$, this gives $p\nmid (mh)^r\mu$.
Thus, $\mu_M\ne0$, so $M$, and hence
$p^{m-1}M$, has rank at least $r$.

The remaining terms in $\HH_{X|Y}(g)$ are two matrices of rank at most one:
\[
  \HH_{X|Y}(g)
  =
  p^{m-1}M
  +m(\nabla_Xp)\left(\nabla_Y(p^{m-1}h)\right)^T
  +m p^{m-1}(\nabla_Xh)(\nabla_Yp)^T,
\]
By the inequality $\rank(A+B)\geq\rank(A)-\rank(B)$, adding a rank-one matrix can decrease the rank by at most one.
Thus the two rank-one terms can decrease the rank of $p^{m-1}M$ by at most two.
Therefore
\[
  \rank\HH_{X|Y}(g)\ge r-2.
\]
\end{proof}

In the following lemma, we observe that under stronger assumptions we get full mixed Hessian rank for powers of irreducible polynomials.

\begin{lemma}[Full-rank criterion for powers]
\label{lem:full-rank-criterion}
Over a field $\F$ of characteristic zero, let $p\in\F[X,Y]$ be irreducible, and define the determinant of $(n+1) \times (n+1)$ matrix
\[
  R_p
  \quad \coloneqq \quad \Det\begin{pmatrix}
    \HH_{X|Y}(p) & \nabla_Xp \\
    \left(\nabla_Yp\right)^T & 0
  \end{pmatrix}.
\]
Suppose that $p\nmid R_p$. Then, for every
\[
  g=p^mh,
  \qquad m\ge2,
  \qquad p\nmid h,
\]
we have
\[
  \rank\HH_{X|Y}(g)=n.
\]
\end{lemma}
\begin{proof}
Let $\vecx=\nabla_Xp$, and $\vecy=\nabla_Yp$. Following again the convention in \cref{def:comm-matrix}, differentiating $g=p^mh$ and factoring out
$p^{m-2}$ gives
\[
\begin{aligned}
  \HH_{X|Y}(g)
  &=p^{m-2}\left(
  \underbrace{m(m-1)h}_{\lambda}\,\vecx\vecy^T
      +p\left(\underbrace{
          mh\HH_{X|Y}(p)+m\vecx\left(\nabla_Yh\right)^T
                   +m\left(\nabla_Xh\right)\vecy^T+p\HH_{X|Y}(h)}_{M}\right)
  \right) \\
  &=p^{m-2}\left(\lambda\vecx\vecy^T+pM\right).
\end{aligned}
\]
The matrix determinant lemma gives
\begin{align}
  \Det\left(\lambda\vecx\vecy^T+pM\right)
  \quad &= \quad -p^{n-1}\Det\begin{pmatrix}
    M & \lambda\vecx \\
    \vecy^T & -p
  \end{pmatrix}
  \quad = \quad p^{n-1}Q,                                             \label{eq:general-full-rank-det}
\end{align}
where
\[
  Q\quad = \quad -\Det\begin{pmatrix}
    M & \lambda\vecx \\
    \vecy^T & -p
  \end{pmatrix}.
\]
Modulo $p$, the term $p\HH_{X|Y}(h)$ in $M$ vanishes. Therefore computing $Q$ modulo $p$ gives
\begin{align*}
Q
&\equiv-\lambda\Det\begin{pmatrix}
  mh\HH_{X|Y}(p)+m\vecx\left(\nabla_Yh\right)^T
  +m\left(\nabla_Xh\right)\vecy^T & \vecx \\
  \vecy^T & 0
\end{pmatrix} \\
&\equiv-\lambda\Det\begin{pmatrix}
  mh\HH_{X|Y}(p) & \vecx \\
  \vecy^T & 0
\end{pmatrix}
\end{align*}
For the second equality, subtract
$m \cdot (\partial h/\partial Y_j)$ times the last column from the $j$th column, and then subtract
$m \cdot (\partial h/\partial X_i)$ times the last row from the $i$th
row. Thus,
\begin{align*}
Q &\equiv-\lambda(mh)^{n-1}R_p \\
&\equiv-m^n(m-1)h^nR_p
\pmod p.
\end{align*}

Since $p$ is irreducible, it is prime in the unique factorization domain $\F[X,Y]$.
Hence, if $p$ does not divide $h$ then it does not divide $h^n$.
Since it divides neither $h^n$ nor $R_p$, it does not divide $h^nR_p$.
Thus, the right-hand side of the preceding congruence is not divisible by $p$, and consequently $Q\ne0$.
By \cref{eq:general-full-rank-det},
\[
  \Det\HH_{X|Y}(g)=p^{n(m-2)+n-1}Q\ne0.
\]
Therefore, $\rank\HH_{X|Y}(g)=n$.
\end{proof}

The full-rank criterion extends from polynomials to rational functions; we prove this generalization in \cref{sec:full-rank-hessian-ext}.

\begin{remark}\label{rmk:2n-lb-via-squaring}
    The assumption $m\geq2$ causes no loss in lower-bound applications.
    Indeed, if $g=p^mh$ with $m\geq1$ and $p\nmid h$, then
    $g^2=p^{2m}h^2$, where $2m\geq2$ and $p\nmid h^2$.
    Thus, \cref{lem:full-rank-criterion,lem:comm-matrix} give $c(g^2)\geq2n$.
    Moreover, any protocol computing $g$ also computes $g^2$ at the same cost by squaring its output, so $c(g^2)\leq c(g)$.
    Hence, $c(g)\geq2n$ as well.
\end{remark}

\section[Framework for probabilistic lower bounds\texorpdfstring{\nobreak\hspace{0.3em}\mbox{\color{white}\fontsize{4}{4}\selectfont It is the real and complex relationship between Alice and Bob}}{}]{Framework for probabilistic lower bounds}
\label{sec:framework}
In this section, we present a general framework for proving algebraic communication lower bounds. The framework is designed by providing sufficient conditions that make the set-recognition problem hard for the communication model. The recipe is similar whether we are working over the reals or complex numbers. However, due to the topological differences between $\R$ and $\C$, the proofs are substantially different. We will start with the case $\F = \R$.

\subsection{Over the real numbers}
\label{sec:framework-reals}
Recall from \cref{lem:full-rank-criterion} that, for an irreducible polynomial $p\in\F[X,Y]$, we define determinant of $(n+1) \times (n+1)$ matrix
\[
R_p=\Det\begin{pmatrix}
    \HH_{X|Y}(p) & \nabla_Xp \\
    \left(\nabla_Yp\right)^T & 0
  \end{pmatrix}.
\]
\begin{theorem}
\label{thm:framework-R}
    Let $S \subseteq \R^n\times\R^n$ be a semialgebraic set and let $F \in \R[X_1,...,X_n,Y_1,...,Y_n]$ satisfying:
    \begin{enumerate}
        \item $F$ is irreducible over $\R$
        \item $\partial S \cap V(F)$ has dimension $2n-1$
        \item $F \nmid R_F$.
    \end{enumerate}
    Then $\PCC_{\R}(S) \geq 2n$
\end{theorem}
\begin{proof}
    Let $\mathcal{C}$ be a probabilistic communication protocol recognizing $S$. By definition,
    $\mathcal{C} $ is a family of deterministic protocols $ \{C_1,...,C_N\}$, where each deterministic protocol $C_i$ is chosen with probability $p_i \geq 0$ such that $\sum_{i \in [N]} p_i = 1$. Let $d$ denote the maximum depth of any deterministic protocol tree in the finite support of $\mathcal{C}$ (i.e., the communication complexity of $\mathcal{C}$).

    For each $C_i$ and a node $v$ in $C_i$, let $g_{i,v}$ be a testing polynomial at node $v$.
    We define $G_i = \prod g_{i,v}$, where the product ranges over the nodes $v$ of $C_i$. We will assume that our protocols are minimal. That is, the protocol has been pruned such that there are no unreachable nodes and any node whose composed testing polynomial is identically zero has been contracted. We will first prove the following lemma about the boundary of $S$. Recall that, if a point lies in the boundary of $S$, namely $\partial S$, then for any ball around $u$, there is at least one point in $S$ and one point outside $S$.

    \begin{claim}\label{claim:containment}
        For any $u \in \partial S$, there exists at least one $i \in [N]$ such that $G_i(u) = 0$, i.e.,
        $$
        \partial S \subseteq \bigcup_{i=1}^N V(G_i).
        $$
    \end{claim} %
    \begin{claimproof}
        Assume, for contradiction, that there exists $u \in \partial S$ such that $G_i(u) \neq 0$, for all $i \in [N]$. Then, no testing polynomial in any $C_i$ vanishes on $u$, meaning $u$ is in the \emph{finite} intersection, (recall that $D(q)$ denotes the nonzero set of $q$)
        $$
        u \in \bigcap_{\substack{C \in \mathcal{C}\\q \in C}}D(q)=:\mathcal{A},
        $$
        where, the intersection ranges over the (nonzero) testing polynomials $q \in C_i$ and across all $C_i \in \mathcal{C}$.

        But, since  $\mathcal{A}$ is a Zariski open set (thus, a Euclidean open set), as it is defined by polynomials, there exists some ball $B_u$ around $u$ such that every testing polynomial across every protocol has constant sign across the elements of $B_u$, and, by extension, all $C_i$ are constant functions on $B_u$.

        To see this, first observe that for every $i \in [N]$, $G_i(u) \neq 0$. This means that, on input $u$, the protocol traverses only along $>0$ or $<0$ tests, and reaches a leaf ($l_i$ in protocol $C_i$). Assume that the output at this leaf is $o_i \in \{\text{accept, reject}\}$. That is, $C_i$ outputs $o_i$ on $u$. Then, by continuity of the polynomials in the protocol, there is a ball of radius $\delta_i$ around $u$, denoted by $B_u(\delta_i)$ , such that every point $z \in B_u(\delta_i)$, $C_i(z) = o_i$. Let $\delta_0  = \text{min}_{i \in [N]} \delta_i$. Then, we know that for every point $z \in B_u(\delta_0)$, and every $i \in [N]$, $C_i(z) = C_i(u)$.

        Now, for any $z \in B_u(\delta_0)$,
        \begin{align}
            \Pr[\mathcal{C} \text{ accepts } z]  & = \sum_{i \in [N]} p_i \times \mathbb{1}[C_i \text{ accepts } z] \nonumber \\
            & = \sum_{i \in [N]} p_i \times \mathbb{1}[C_i \text{~accepts } u] \nonumber \\
            & = \Pr[\mathcal{C} \text{~accepts } u] \label{eq:ball-around-u},
        \end{align}
        where $\mathbb{1}[C_i \text{ accepts } z]$ is $1$ if $C_i$ accepts $z$ and $0$ otherwise. The second equality comes from our argument above and the fact that $z \in B_u(\delta_0)$.

        But, $u$ is in the boundary of $S$. This means that in any ball around $u$, and, hence specifically in  $B_u(\delta_0)$, we have two points $z_{\text{in}} \in S$ and $z_{\text{out}} \notin S$. Thus, we know that $\Pr[\mathcal{C} \text{ accepts } z_{\text{in}} ] \geq 2/3$ and $\Pr[\mathcal{C} \text{ accepts } z_{\text{out}} ] \leq 1/3$.
        Thus, $2/3 \leq Pr[C_i \text{ accepts } z_{\text{in}}] = Pr[ C_i  \text{ accepts } z_{\text{out}}] \leq 1/3$, which is a contradiction.
    \end{claimproof}
    We use the above claim and the following fact to complete the proof of \Cref{thm:framework-R}.
    \begin{fact}\label{fact:dim-red}\cite{Henk2008}
        If $F$ is irreducible, $\dim(V(F)) = 2n-1$ and $F \nmid G$ then
        $$
        \dim(V(F) \cap V(G)) \leq 2n-2
        $$
    \end{fact}
    At an intuitive level, the above fact says that, given a set of zeroes of an irreducible polynomial $F$, if we impose an additional non-trivial constraint $G = 0$ on the set, then the codimension of the intersection must be at least 1, i.e., the dimension must drop by at least 1.  Note that, the assumption $\dim V(F)=2n-1$ ensures that the real zero set of $F$ contains a genuine hypersurface piece.\footnote{In particular, it excludes polynomials such as $x^2+y^2$, whose real zero set has smaller dimension.}

    Using the above fact, we will first show that $F$ must divide some $G_i$. Suppose it does not. Then, for any $i \in [N]$, not all the roots of $F$ can be roots of $G_i$. Thus, if $F$ did not divide any $G_i$, then we would have
    $\dim\left(V(F) \cap V(G_i)\right) \leq 2n-2$ for all $i\in[N]$ and so
    \begin{equation}\label{eq:dimension-cup-cap}
    \dim\left( \bigcup_{i=1}^N\left(V(F) \cap V(G_i)\right) \right) \leq 2n-2.
    \end{equation}
    By \Cref{claim:containment}, we have $\partial S \cap V(F) \subseteq \bigcup_{i=1}^N\left(V(F) \cap V(G_i)\right),$
    so \cref{eq:dimension-cup-cap} implies that
    \begin{equation}\label{eq:subset}
    \dim(\partial S \cap V(F)) \leq \dim\left( \bigcup_{i=1}^N(V(F) \cap V(G_i)) \right) \leq 2n-2,
    \end{equation}
    which contradicts Assumption $\mathit{ 2 }$.

    Thus, $F$ must divide some $G_i$. Since $F$ is irreducible (in particular prime), it must divide some individual testing polynomial, $G$, occurring in the tree $C_i$. That is, $G = F^m \cdot H$, where $m\geq 1$ and $F \nmid H$. The polynomial $F$ being irreducible also means that $G^2 = F^{2m}\cdot H^2 = F^{m'}\cdot H^2$, where $m'\geq 2$ and $F \nmid H^2$. From our assumption in \Cref{thm:framework-R}, we also know that $F \nmid R_F$.

    We now apply \Cref{lem:full-rank-criterion} to $G^2$ and obtain that $c(G^2) \geq 2 \rank(\HH_{X|Y}(G^2)) = 2 n$. Next, notice that $c(G) \geq c(G^2)$: suppose $c(G) = s$, it implies that $G(X,Y) = g(a_1(X), \ldots, a_{s_1}(X), b_1(Y), \ldots, b_{s_2}(Y))$, where $s_1+s_2 =s$. Then, $G^2(X,Y)$ is computed by simply squaring $g$. Thus, $c(G^2) \leq s$. Together with the lower bound on $c(G^2)$, we get $c(G)\geq 2n$ and hence, the communication complexity of recognizing $S$, i.e. the depth of the protocol, which is at least $c(G)$,  is at least $2n$.
\end{proof}

\begin{remark}
    The condition that $\dim(\partial S \cap V(F)) = 2n-1$ is crucial and cannot be made weaker. Indeed, if $\dim(\partial S \cap V(F)) \leq 2n-2$ then \cref{eq:subset} would not be a contradiction. In other words, we need dimension $2n-1$ so that the boundary is too large to lie inside $V(F)\cap V(G)$ unless $F\mid G$.
\end{remark}

A weaker mixed Hessian assumption gives a correspondingly weaker lower bound.
If some $r\times r$ minor of $\HH_{X|Y}(F)$ is not divisible by $F$, then \cref{lem:general-rank} can be used in the final step of the preceding proof to obtain a lower bound of $2(r-2)$.

\begin{theorem}
\label{thm:framework-R-mixed-hessian}
Let $S\subseteq\R^n\times\R^n$ be a semialgebraic set and let
$F\in\R[X_1,\ldots,X_n,Y_1,\ldots,Y_n]$, and let $1\leq r\leq n$. Suppose that
\begin{enumerate}
  \item $F$ is irreducible over $\R$;
  \item $\partial S\cap V(F)$ has dimension $2n-1$; and
  \item some $r\times r$ minor of $\HH_{X|Y}(F)$ is not divisible by $F$.
\end{enumerate}
Then $\PCC_{\R}(S)\geq2(r-2)$.
\end{theorem}

We now derive the main result stated in the introduction by specializing $S$ in \cref{thm:framework-R} to the threshold set $S(F)$.

\getkeytheorem{realthresholdtheorem}
\begin{proof}
    Let $(z_1,...,z_{2n}) = (X,Y)$. It suffices to verify Conditions 2 and 3 of \cref{thm:framework-R}. Condition 3 follows immediately: if $F \mid R_F$, then $V(F) \subseteq V(R_F)$, contradicting $F(u)=0$ and $R_F(u)\neq 0$.

    It remains to verify Condition 2, namely,
    \[
        \dim(\partial S(F) \cap V(F)) = 2n-1.
    \]

    Since $R_F(u) \neq 0$, neither the last row nor the last column is zero. Hence,
    $$
    \nabla_XF(u)\neq 0 \qquad \text{and} \qquad \nabla_YF(u) \neq 0,
    $$
    so $\nabla F(u) \neq 0$.
    We may therefore fix $i \in [2n]$ such that $\partial F/\partial z_i(u) \neq 0$. By continuity, there exists an open ball $B_u$ around $u$ such that $\partial F/\partial z_i$ has constant sign on $B_u$.
    Since $u$ is non-singular, by the Implicit Function Theorem, intersecting $V(F)$ with a sufficiently small ball around $u$ preserves its dimension. Hence, shrinking $B_u$ further if necessary, we get by \cref{prop:real_hypersurface},
    \begin{equation}\label{eq:dim-F-Bu}
        \dim(V(F) \cap B_u) = 2n-1.
    \end{equation}

    Now, fix $z = (z_1,...,z_{2n}) \in V(F) \cap B_u$.
    Since we assumed $\partial F / \partial z_i$ has constant nonzero sign on $B_u$, we get that $F$ is strictly monotone in the $z_i$-direction. Thus, for sufficiently small $\varepsilon>0$, we have
    $$
    F(z_1,...,z_i-\varepsilon,...,z_{2n}) < F(z_1,...,z_i,...,z_{2n})=0 < F(z_1,...,z_i+\varepsilon,...,z_{2n})
    $$
    if increasing and with the inequalities reversed if decreasing. In either case the two points
    $$
    (z_1,...,z_i-\varepsilon,...,z_{2n}) \qquad \text{and} \qquad (z_1,...,z_i+\varepsilon,...,z_{2n})
    $$
    lie on opposite sides of $S(F)$ and so we conclude that $z\in \partial S(F)$. Thus, since $z$ was arbitrary, we have
    \begin{equation}\label{eq:F-cap-B_u-in-SF}
        V(F) \cap B_u \subseteq \partial S(F)
    \end{equation}

    Finally, \cref{eq:F-cap-B_u-in-SF} together with \cref{eq:dim-F-Bu} gives us
    $$
    2n-1 = \dim(V(F) \cap B_u) \leq \dim(\partial S(F) \cap V(F)) \leq 2n-1.
    $$
    Hence, $\dim(\partial S(F) \cap V(F)) = 2n-1$, as desired.
\end{proof}

\begin{remark}
    The same proof works for the stricter threshold set,
    $$
    S_{< 0}(F) = \{(X,Y) \in \R^n\times\R^n \mid F(X,Y) < 0\}.
    $$
    All results regarding weak threshold sets thus remain true after replacing it with a strict threshold set.
\end{remark}

\subsection{Over the complex numbers}

We now prove the complex analogue of the real framework. The real framework is formulated using the Euclidean boundary, whereas, over $\C$, we express the corresponding condition using Zariski closure.

\begin{theorem}\label{thm:framework-C}
    Let $S \subseteq \C^n\times\C^n$ be a constructible set of dimension $2n-1$ and let $F \in \C[X_1,...,X_n,Y_1,...,Y_n]$ satisfy
    \begin{enumerate}
        \item $F$ is irreducible over $\C$
        \item $V(F) \subseteq \overline{S}$, where $\overline{S}$ is the Zariski closure of $S$.
        \item $F \nmid R_F$.
    \end{enumerate}
    Then, $\PCC_{\C}(S) \geq 2n$.
\end{theorem}
\begin{proof}

    Let $U = V(F)$ and note that almost every point of $U$ lies in $S$.\footnote{$U$ is irreducible and has dimension $2n-1$. Since $U \subseteq \overline{S}$ and $\dim \overline{S}=\dim S=2n-1$, $U$ is an irreducible component of $\overline{S}$ and so $S\cap U$ contains a dense Zariski-open subset of $U$, since $S$ is constructible.}

    Let $\mathcal{C}=\{C_1,...,C_N\}$ be a probabilistic communication protocol recognizing $S$, and let $\mu_1, \ldots, \mu_N$ be the corresponding probability distribution. The following claim identifies a single protocol $C_k$ that is correct on a constant fraction of the points both inside and outside $U$.

    \begin{claim}\label{clm:select-protocol}
        There exists $C_k\in\mathcal{C}$ that accepts at least $1/3$ of the points in $U$ and rejects at least $1/3$ of the points outside $U$.
    \end{claim}
    \begin{claimproof}
        For each $j\in[N]$, let $A_j$ be the fraction of points in $U$ accepted by $C_j$ and let $B_j$ be the fraction of points outside $U$ rejected by $C_j$. %
        By the correctness of the probabilistic protocol,\footnote{%
        By the previous footnote, almost every point of $U$ belongs to $S$. Moreover, almost every point outside $U$ lies outside $S$.}
        \[
            \sum_{j\in[N]}\mu_jA_j\geq\frac{2}{3}
            \qquad\text{and}\qquad
            \sum_{j\in[N]}\mu_jB_j\geq\frac{2}{3}.
        \]
        Therefore,
        \[
            \sum_{j\in[N]}\mu_j(A_j+B_j)\geq\frac{4}{3},
        \]
        so there exists $k\in[N]$ such that $A_k+B_k\geq4/3$. Since $A_k,B_k\leq1$, both $A_k$ and $B_k$ are at least $1/3$.
    \end{claimproof}

    Let $T_{1},...,T_{m}$ be the polynomials along the path $\pi$ in $C_k$, in which all decisions are $T_i(X,Y)\neq0$. Take the product
    $$
    G_k = \prod_{i=1}^{m} T_{i}.
    $$
    First, we will show that $\pi$ as defined above is a rejecting path in $C_k$. We will use this fact to deduce that $F | G_k$. Once we have that $F|G_k$, then the proof is similar to the proof of \Cref{thm:framework-R}.

    \begin{claim}
        \label{clm:reject-path}
    Let $\pi$ and $G_k$ as defined above. Then any $u \in D(G_k)$, i.e., any input that ends up at the leaf of $\pi$, is rejected by $C_k$.
    \end{claim}

    \begin{claimproof}
        Let $\pi'$ be some path other than $\pi$ and let $T_1',\cdots,T_{m'}'$ be the testing polynomials encountered along $\pi'$. Since $\pi'$ differs from $\pi$, the points, $(X',Y')$, which follow $\pi'$ must have $T_i'(X',Y')=0$ for some $i \in [{m'}]$.
        Since $T'_i$ is nonzero, $V(T'_i)$ is either empty or a finite union of hypersurfaces of dimension $2n-1$, by \cref{prop:real_hypersurface}. Hence the set of inputs following $\pi'$, being contained in $V(T'_i)$, has dimension at most $2n-1$.

        Thus, $\pi$ is the only path for which a full dimensional subset follows. Furthermore, the space of points we wish to reject, $(\C^n \times \C^n) \setminus S$, is full dimensional, since $S$ has dimension $2n-1$ and constructible. By  \Cref{clm:select-protocol}, $1/3$ of points of $(\C^n \times \C^n) \setminus S$ are correctly rejected and this can only be the case if $\pi$ is a rejecting path.
    \end{claimproof}

    \begin{claim}\label{claim:F-div-Gk}
        $F \mid G_k$
    \end{claim}
    \begin{claimproof}
        Suppose, for contradiction, that $F \nmid G_k$. As $F$ is irreducible, using \Cref{fact:dim-red}, we get that
        $$
        \dim(E) \leq 2n-2, \qquad \text{where} \qquad E=U \cap V(G_k).
        $$

        Furthermore, for every $u \in U \setminus E$ we have $G_k(u) \neq 0$ and so $C_k$ rejects $u$. Thus
        $$
        B = \{u \in U \mid C_k(u)= \text{accept} \} \subseteq E.
        $$
        Hence, $C_k$ only accepts points from a dimension $2n-2$ set, which is a zero fraction of the points of $U$. But $C_k$ was chosen to accept $1/3$ of points from $U$. This is a contradiction.
    \end{claimproof}

    We can now write $G_k = F^m \cdot H$, where $m\geq 1$ and $F \nmid H$. We know that $F$ is irreducible. This means that $G_k^2 = F^{2m}\cdot H^2 = F^{m'}\cdot H^2$, where $m'\geq 2$ and $F \nmid H^2$. From our assumption in \Cref{thm:framework-C}, we also know that $F \nmid R_F$. The rest of the proof is identical to the proof of \Cref{thm:framework-R}, which shows that $c(G_k)\geq 2n$ and hence, the communication complexity of recognizing $S$ is at least $2n$, i.e. the depth of the protocol, which is at least $c(G_k)$,  is at least $2n$.
\end{proof}

As in the real case, a weaker assumption on the mixed Hessian gives a correspondingly weaker lower bound.
Using a non-divisible $r\times r$ minor together with \cref{lem:general-rank} yields the following complex analogue of \cref{thm:framework-R-mixed-hessian}.

\begin{theorem}[Mixed Hessian framework over $\C$]
\label{thm:framework-C-mixed-hessian}
Let $S\subseteq\C^n\times\C^n$ be a constructible set of dimension
$2n-1$, let $F\in\C[X_1,\ldots,X_n,Y_1,\ldots,Y_n]$, and let $1\leq r\leq n$. Suppose that
\begin{enumerate}
  \item $F$ is irreducible over $\C$;
  \item $V(F)\subseteq\overline{S}$, where $\overline{S}$ is the Zariski closure of $S$; and
  \item some $r\times r$ minor of $\HH_{X|Y}(F)$ is not divisible by $F$.
\end{enumerate}
Then
\[
  \PCC_{\C}(S)\geq2(r-2).
\]
\end{theorem}

In the complex case we do not have threshold sets, but the natural analogue would be the hypersurface-membership problem:

\getkeytheorem{complexhypersurfacetheorem}
\begin{proof}
    Clearly $V(F)$ is a constructible set of dimension $2n-1$, so we need only verify that Conditions 2 and 3 of \cref{thm:framework-C} are satisfied. First observe that Condition 2 is trivial since $V(F)$ is already Zariski closed, so $V(F) = \overline{V(F)}$.  Next, Condition 3, that $F \nmid R_F$, is also easy to see. If $F \mid R_F$ then $V(F) \subseteq V(R_F)$ which contradicts the existence of $u$.
\end{proof}

\section{New probabilistic lower bounds}
\label{sec:new-lower-bounds}

We now apply the framework developed in the previous section to prove probabilistic communication lower bounds for several natural set-recognition problems over both $\R$ and $\C$.

\subsection{The \texorpdfstring{$\varepsilon$}{ε}-equality set}

For $\varepsilon>0$, define\footnote{The polynomial $\EQ_n$ exactly captures equality of vectors over $\R$, but not over $\C$ when $n\geq2$. Nevertheless, we retain the notation for consistency.}
\begin{align}
  \EQ_n(X,Y)=\sum_{i=1}^n\left(X_i-Y_i\right)^2,
  \qquad
  \EQ_{n,\varepsilon}=\EQ_n-\varepsilon \label{eq:eps-eq}
\end{align}
and
\[
  S(\EQ_{n,\varepsilon})
  =\left\{(X,Y)\in\R^n\times\R^n \mid \EQ_{n,\varepsilon}(X,Y)\leq0\right\}.
\]
Thus, $S(\EQ_{n,\varepsilon})$ consists of pairs of vectors whose squared
Euclidean distance is at most $\varepsilon$. We first prove that the defining polynomial is irreducible.

\begin{lemma}[Irreducibility of $\EQ_{n,\varepsilon}$]
\label{lem:epsilon-equality-irreducible}
For every $n\geq2$ and $\varepsilon>0$, the polynomial $\EQ_{n,\varepsilon}(X,Y)$ defined
in \cref{eq:eps-eq} is irreducible over both $\R$ and $\C$.
\end{lemma}

\begin{proof}
Set $d_i=X_i-Y_i$. Since the map $(X_i, Y_i) \mapsto (X_i - Y_i, Y_i)$ is an invertible linear change of variables, it suffices to prove that
\[
  q(d)=\sum_{i=1}^n d_i^2-\varepsilon
\]
is irreducible over $\C$. Its homogenization is
\[
  Q(d,z)=\sum_{i=1}^n d_i^2-\varepsilon z^2.
\]
Since $Q$ is a quadratic form of rank $n+1\geq3$, it is irreducible over
$\C$ (see \cite{Tverberg1964}). If $q$ had a non-trivial
factorization, homogenizing its factors would give a non-trivial
factorization of $Q$. Hence $q$, and therefore $\EQ_{n,\varepsilon}$, is irreducible over
$\C$, and consequently over $\R$.

\end{proof}

We now prove the communication complexity of the real threshold set $S(\EQ_{n,\varepsilon})$ and of the corresponding complex hypersurface.

\begin{theorem}[note={Communication complexity of $\EQ_{n,\varepsilon}$},
  label=thm:epsilon-equality]
Let $n\geq2$ and $\varepsilon>0$, and set
\[
  \EQ_n(X,Y)=\sum_{i=1}^n\left(X_i-Y_i\right)^2,
  \qquad
  \EQ_{n,\varepsilon}=\EQ_n-\varepsilon,
\]
\[
  S(\EQ_{n,\varepsilon})
  =\left\{(X,Y)\in\R^n\times\R^n \mid \EQ_{n,\varepsilon}(X,Y)\leq0\right\}.
\]
Then,
\[
  \PCC_{\R}\left(S(\EQ_{n,\varepsilon})\right)
  =\DCC_{\R}\left(S(\EQ_{n,\varepsilon})\right)
  =2n.
\]
Moreover,
\[
  \PCC_{\C}\left(V(\EQ_{n,\varepsilon})\right)
  =\DCC_{\C}\left(V(\EQ_{n,\varepsilon})\right)
  =2n.
\]
\end{theorem}

\begin{proof}
By \cref{lem:epsilon-equality-irreducible}, the polynomial $\EQ_{n,\varepsilon}$ is
irreducible over both $\R$ and $\C$.
Let $u\in\R^n\times\R^n\subseteq\C^n\times\C^n$ be any
point satisfying
\[
  X-Y=\left(\sqrt{\varepsilon},0,\ldots,0\right).
\]
Then, $\EQ_{n,\varepsilon}(u)=0$. Writing $d=X-Y$, we have
\[
  \nabla_X\EQ_{n,\varepsilon}=2d,
  \qquad
  \nabla_Y\EQ_{n,\varepsilon}=-2d,
  \qquad
  \HH_{X|Y}(\EQ_{n,\varepsilon})=-2I_n,
\]
and hence
\[
  R_{\EQ_{n,\varepsilon}}(u)=4(-2)^{n-1}d^Td
        =4(-2)^{n-1}\varepsilon\neq0.
\]
Therefore, \cref{thm:framework-R-threshold} gives
$\PCC_{\R}\left(S(\EQ_{n,\varepsilon})\right)\geq2n$.
The same point $u$ also satisfies the hypotheses of
\cref{thm:framework-C-hypersurface}, so
\[
  \PCC_{\C}\left(V(\EQ_{n,\varepsilon})\right)\geq2n.
\]

For the matching upper bounds, Alice and Bob broadcast their coordinates,
using $2n$ messages in total, after which the testing polynomial evaluates
$\EQ_{n,\varepsilon}(X,Y)$. Over $\R$, the protocol accepts when this value is nonpositive;
over $\C$, it accepts when the value is zero. Thus, both deterministic
protocols have depth $2n$, completing the proof.
\end{proof}

\subsection{The \texorpdfstring{$\varepsilon$}{ε}-hypercube set}

For $\varepsilon>0$, define
\begin{equation}
  \HC_n(X,Y)
  =\sum_{i=1}^n\left(\left(X_i-Y_i\right)^2-1\right)^2,
  \qquad
  \HC_{n,\varepsilon}=\HC_n-\varepsilon
  \label{eq:eps-hypercube}
\end{equation}
and
\[
  S(\HC_{n,\varepsilon})
  =\left\{(X,Y)\in\R^n\times\R^n \mid \HC_{n,\varepsilon}(X,Y)\leq0\right\}.
\]

The set $S(\EQ_{n,\varepsilon})$ can be viewed as recognizing whether a point $Z\in\R^n$ is close to the origin, when $Z$ is distributed between the two parties and presented as $Z=X-Y$.
We can generalize this to sets recognizing whether a distributed point $Z=X-Y$ is close to some point in a finite structured set.
Taking this structured set to be the vertices $\{-1,1\}^n$ of the Boolean hypercube, we obtain $S(\HC_{n,\varepsilon})$, which we call the $\varepsilon$-hypercube problem.

We first prove that its defining polynomial is
irreducible.

\begin{lemma}[Irreducibility of $\HC_{n,\varepsilon}$]
\label{lem:p_epsilon-irreducible}
For every  $n \geq 3$ and $\varepsilon\geq0$, the polynomial $\HC_{n,\varepsilon}(X,Y)$ defined
in \cref{eq:eps-hypercube} is irreducible over $\C$.
\end{lemma}
\begin{remark}
  For $n=2$, it can be proved that $\HC_{2,\varepsilon}$ is irreducible over $\R$ whenever $\varepsilon \neq 1$.
\end{remark}
\begin{proof}[Proof of \Cref{lem:p_epsilon-irreducible}]
  Let $d_i=X_i-Y_i$. Since this is an invertible linear change of variables, it suffices to prove irreducibility in the variables $d_1,\ldots,d_n$.
  The highest-degree homogeneous part of
  \[
    \HC_{n,\varepsilon}=\sum_{i=1}^n(d_i^2-1)^2-\varepsilon
  \]
  is $A=\sum_{i=1}^n d_i^4$, which is irreducible over $\C$ when $n \geq 3$ (see \cite{Tverberg1964}).
  If $\HC_{n,\varepsilon}=gh$ were a non-trivial factorization, then the highest-degree homogeneous parts of $g$ and $h$ would give a non-trivial factorization of $A$, a contradiction.
  Thus $\HC_{n,\varepsilon}$ is irreducible over $\C$, and hence over $\R$.
\end{proof}

Having proved the irreducibility of $\HC_{n,\varepsilon}$, we compute the mixed Hessian needed to
apply the lower-bound framework.

\begin{lemma}[Mixed Hessian of $\HC_{n,\varepsilon}$]
\label{lem:hessian-p}
Let $d_i=X_i-Y_i$. Then
\[
  \HH_{X|Y}(\HC_{n,\varepsilon})
  =-\operatorname{diag}(12d_1^2-4,\ldots,12d_n^2-4).
\]
In particular,
\[
  \Delta_{n,\varepsilon}\coloneqq\Det\HH_{X|Y}(\HC_{n,\varepsilon})
  =(-1)^n\prod_{i=1}^n(12d_i^2-4).
\]
\end{lemma}

\begin{proof}
Each summand of $\HC_{n,\varepsilon}$ depends on only one $d_i$, so the mixed
derivative vanishes when $i\ne j$. For $i=j$,
\begin{align}
  \frac{\partial \HC_{n,\varepsilon}}{\partial X_i}
  =4d_i(d_i^2-1),
  \qquad
  \frac{\partial^2\HC_{n,\varepsilon}}{\partial X_i\partial Y_i}
  =-4(3d_i^2-1)=-(12d_i^2-4). \label{eq:gradient-eps-hypercube}
\end{align}
The two formulas follow.
\end{proof}

We now determine the communication complexity of the real threshold set
$S(\HC_{n,\varepsilon})$ and of the corresponding complex hypersurface.
\begin{theorem}[note={Communication complexity of $\HC_{n,\varepsilon}$},
  label=thm:epsilon-hypercube]
Let $\varepsilon > 0$ and $n\geq3$, and set
\[
  \HC_n(X,Y)
  =\sum_{i=1}^n\left(\left(X_i-Y_i\right)^2-1\right)^2,
  \quad
  S(\HC_{n,\varepsilon})
  =\left\{(X,Y)\in\R^n\times\R^n \mid \HC_{n,\varepsilon}(X,Y)\leq0\right\},
\]
where $\HC_{n,\varepsilon}=\HC_n-\varepsilon$.
Then
\[
  \PCC_{\R}\left(S(\HC_{n,\varepsilon})\right)
  =\DCC_{\R}\left(S(\HC_{n,\varepsilon})\right)
  =2n.
\]
Moreover,
\[
  \PCC_{\C}\left(V(\HC_{n,\varepsilon})\right)
  =\DCC_{\C}\left(V(\HC_{n,\varepsilon})\right)
  =2n.
\]
\end{theorem}

\begin{proof}
By \cref{lem:p_epsilon-irreducible}, the polynomial $\HC_{n,\varepsilon}$ is irreducible over the relevant field under the stated assumptions.
Choose $u\in\R^n\times\R^n\subseteq\C^n\times\C^n$ such that, with $d=X-Y$,
\[
  d_1^2=1+\sqrt{\varepsilon},
  \qquad
  d_2=\cdots=d_n=1.
\]
Then $\HC_{n,\varepsilon}(u)=0$. Let $v=\nabla_X\HC_{n,\varepsilon}(u)$. From \cref{eq:gradient-eps-hypercube}
\[
  v_i=4d_i(d_i^2-1),
\]
we have $v_1=4d_1\sqrt{\varepsilon}\neq0$ and $v_i=0$ for $i\geq2$.
By \cref{lem:hessian-p}, the last $n-1$ diagonal entries of $\HH_{X|Y}(\HC_{n,\varepsilon})(u)$ are all $-8$.
Moreover, $\nabla_Y\HC_{n,\varepsilon}=-\nabla_X\HC_{n,\varepsilon}$, and hence
\[
  R_{\HC_{n,\varepsilon}}(u)
  =\Det\begin{pmatrix}
    \HH_{X|Y}(\HC_{n,\varepsilon})(u) & v \\
    -v^T & 0
  \end{pmatrix}
  =v_1^2(-8)^{n-1}\neq0.
\]
Applying \cref{thm:framework-R-threshold} gives
\[
  \PCC_{\R}\left(S(\HC_{n,\varepsilon})\right)\geq2n,
\]
while \cref{thm:framework-C-hypersurface} gives
\[
  \PCC_{\C}\left(V(\HC_{n,\varepsilon})\right)\geq2n.
\]

For the upper bound, Alice broadcasts $h_i=X_i$ and Bob broadcasts $h_{n+i}=Y_i$ for every $i\in[n]$.
The final testing polynomial is
\[
  \sum_{i=1}^n\left(\left(h_i-h_{n+i}\right)^2-1\right)^2-\varepsilon.
\]
Over $\R$, the protocol accepts when this value is nonpositive.
Over $\C$, it accepts when the value is zero.
This gives deterministic protocols of depth $2n$ for $S(\HC_{n,\varepsilon})$ and $V(\HC_{n,\varepsilon})$, respectively.
Together with the two lower bounds, this proves both equalities.
\end{proof}

\subsection{Inner product}\label{sec:inner-product}

Grigoriev proved a lower bound of $2(n-3)$ for the inner product set-recognition problems~\cite[Proposition~3.1]{Gri2008}.
Using our framework together with the mixed Hessian rank bound of \cref{lem:general-rank} (instead of \cref{lem:full-rank-criterion}), we improve this bound to be $2(n-2)$.

Recall that
\[
  \IP_n(X,Y)=\sum_{i=1}^n X_iY_i.
\]
Over $\R$, we consider the threshold set $S(\IP_n)$, while over $\C$, we consider the zero set $V(\IP_n)$.

\begin{theorem}[Probabilistic lower bound for inner product]
\label{thm:inner-product-set-recognition}
For every $n\geq3$,
\[
  \PCC_{\R}\left(S(\IP_n)\right)\geq2(n-2)
  \qquad\text{and}\qquad
  \PCC_{\C}\left(V(\IP_n)\right)\geq2(n-2).
\]
Consequently, the same lower bound holds for
$\DCC_{\R}\left(S(\IP_n)\right)$ and
$\DCC_{\C}\left(V(\IP_n)\right)$.
\end{theorem}

\begin{proof}
The polynomial $\IP_n$ is irreducible over both $\R$ and $\C$.
Over $\R$, consider the point
$u=((1,0,\ldots,0),0)\in V(\IP_n)$.
Since $\nabla \IP_n(u)=(0,(1,0,\ldots,0))\ne0$,
\cref{prop:real_hypersurface} gives $\dim V(\IP_n)=2n-1$.
\begin{claim}
$\partial S(\IP_n)=V(\IP_n)$.
\end{claim}
\begin{claimproof}
First, let $(X,Y)\notin V(\IP_n)$. Then $\IP_n(X,Y)$ is either strictly positive, or strictly negative. By continuity, the same strict inequality holds throughout some open ball around $(X,Y)$. This ball is therefore either contained in $S(\IP_n)$, or disjoint from it, so $(X,Y)\notin\partial S(\IP_n)$. Hence
\[
  \partial S(\IP_n)\subseteq V(\IP_n).
\]

For the reverse inclusion, fix $(X,Y)\in V(\IP_n)$ and an arbitrary open ball $B_r(X,Y)$. We construct within this ball one point where $\IP_n$ is negative and another where it is positive.

Suppose first that $X\ne0$. For every $t>0$,
\[
  \IP_n(X,Y+tX)=t\sum_{i=1}^n X_i^2>0,
  \qquad
  \IP_n(X,Y-tX)=-t\sum_{i=1}^n X_i^2<0,
\]
where we used the bilinearity of $\IP_n$ and the fact that $\IP_n(X,Y)=0$. Choosing $t<r/\lVert X\rVert$ places both points in $B_r(X,Y)$.

If $X=0$ and $Y\ne0$, the same argument applies after perturbing the first input:
\[
  \IP_n(tY,Y)=t\sum_{i=1}^n Y_i^2>0,
  \qquad
  \IP_n(-tY,Y)=-t\sum_{i=1}^n Y_i^2<0.
\]
Again, both points lie in $B_r(X,Y)$ for sufficiently small $t>0$.

Finally, if $X=Y=0$, then for sufficiently small $t>0$ the points
$((t,0,\ldots,0),(t,0,\ldots,0))$ and
$((t,0,\ldots,0),(-t,0,\ldots,0))$ lie in $B_r(0,0)$, while
\[
  \IP_n((t,0,\ldots,0),(t,0,\ldots,0))=t^2>0,
  \qquad
  \IP_n((t,0,\ldots,0),(-t,0,\ldots,0))=-t^2<0.
\]
Thus every open ball centered at $(X,Y)$ meets both $S(\IP_n)$ and its complement. Therefore $(X,Y)\in\partial S(\IP_n)$, proving the reverse inclusion.
\end{claimproof}
We have
\[
  \HH_{X|Y}(\IP_n)=I_n.
\]
Hence $\Det\HH_{X|Y}(\IP_n)=1$, which is not divisible by $\IP_n$.
Together with the claim and $\dim V(\IP_n)=2n-1$, this verifies the hypotheses of \cref{thm:framework-R-mixed-hessian} with $r=n$. Therefore,
\[
  \PCC_{\R}\left(S(\IP_n)\right)\geq2(n-2).
\]

Over $\C$, \cref{prop:real_hypersurface} gives $\dim V(\IP_n)=2n-1$.
Moreover, $V(\IP_n)$ is Zariski closed and hence equals its Zariski closure.
Thus, the hypotheses of \cref{thm:framework-C-mixed-hessian} are satisfied with $r=n$, and
\[
  \PCC_{\C}\left(V(\IP_n)\right)\geq2(n-2).
\]
\end{proof}

We now return to the bilinear-form threshold and zero-set problems introduced in \cref{sec:bilinear-forms}.

\begin{corollary}[Communication bounds for bilinear forms]
\label{thm:bilinear-form-bounds}
Let $A\in\R^{n\times n}$ and $r=\rank_{\R}(A)\geq1$.
Then
\[
  2r-4
  \leq \PCC_\R\left(S(\BF_{n,A})\right)
  \leq \DCC_\R\left(S(\BF_{n,A})\right)
  \leq 2r.
\]
For the complex case, let $A\in\C^{n\times n}$ and $r=\rank_{\C}(A)\geq1$.
Then
\[
  2r-4
  \leq \PCC_\C\left(V(\BF_{n,A})\right)
  \leq \DCC_\C\left(V(\BF_{n,A})\right)
  \leq 2r.
\]
\end{corollary}

\begin{proof}
The upper bounds follow from \cref{obs:bilinear-inner-product-equivalence} and the protocol in which Alice and Bob each broadcast their $r$ coordinates.
For $r\leq2$, the lower bounds are immediate.
For $r\geq3$, \cref{thm:inner-product-set-recognition} gives the lower bound $2(r-2)$ in both cases, and the result follows from \cref{obs:bilinear-inner-product-equivalence}.
\end{proof}

\subsection{The \texorpdfstring{$\varepsilon$}{ε}-SI set}

The set-intersection problem asks whether two input vectors, viewed as sets, share a common element; that is, whether $X_i=Y_j$ for some
$i,j\in[n]$.
Grigoriev proved probabilistic communication lower bounds of
$n$ over $\R$ and $n/2$ over $\C$ for this problem
(see \cite[Corollary~5.5]{Gri2008}).
We consider the following perturbed variant.
Let
\begin{equation}
  \SI_n(X,Y)=\prod_{i=1}^n\prod_{j=1}^n(X_i-Y_j),
  \qquad
  \SI_{n,\varepsilon}(X,Y)=\SI_n(X,Y)-\varepsilon,
  \label{eq:epsilon-si}
\end{equation}
where $\varepsilon>0$, and define
\[
  S(\SI_{n,\varepsilon})
  =\setdef{(X,Y)\in\R^n\times\R^n}{\SI_{n,\varepsilon}(X,Y)\leq0}.
\]
We first prove the irreducibility of the  $\epsilon$-SI polynomial \footnote{We will show that $\SI_{n,\varepsilon}$ is irreducible over all $\epsilon \neq 0$. However, we will use it when $\varepsilon > 0$.}.

\begin{lemma}
\label{lem:epsilon-si-irreducible}
For every $n\geq2$ and $\varepsilon \in \R \backslash \{0\}$, the polynomial $\SI_{n,\varepsilon}$ defined in \cref{eq:epsilon-si} is irreducible over both $\R$ and $\C$.
\end{lemma}

\begin{proof}
It suffices to prove irreducibility over $\C$.
Set $d=n^2$.
The homogenization of $\SI_{n,\varepsilon}$ is
\[
  \SI_n(X,Y)-\varepsilon Z^d.
\]
The polynomials $X_i-Y_j$ are irreducible, and no two are nonzero constant multiples of one another. Hence, every irreducible factor of $\SI_n/\varepsilon$ occurs with multiplicity one, whereas its multiplicity in a $q$th power would be divisible by $q$; thus, $\SI_n/\varepsilon$ is not a $q$th power in $\C(X,Y)$ for any prime $q\mid d$.
Therefore, by the binomial irreducibility criterion, the homogenization is irreducible when regarded as a polynomial in $Z$ over $\C(X,Y)$.
By Gauss's lemma~\cite[Section~6.2]{GathenGerhard2013}, the homogenization is therefore irreducible in $\C[X,Y,Z]$.
If $\SI_{n,\varepsilon}$ admitted a non-trivial factorization over $\C$, homogenizing its factors would give a non-trivial factorization of $\SI_n(X,Y)-\varepsilon Z^d$.
Hence $\SI_{n,\varepsilon}$ is irreducible over $\C$, and therefore also over $\R$.
\end{proof}

For the mixed Hessian rank we use the following identity of Borchardt.

\begin{proposition}[Borchardt's identity \protect{\cite[Equation~2]{borchardt1857}}]
\label{prop:borchardt}
Let $C$ and $C^{(2)}$ be the $n\times n$ matrices defined by
\[
  C_{ij}=\frac{1}{X_i-Y_j},
  \qquad
  C^{(2)}_{ij}=\frac{1}{(X_i-Y_j)^2}.
\]
Then
\[
  \Det C^{(2)}
  \quad = \quad \Det(C)\Perm(C)
  \quad =\quad \frac{\displaystyle
      \prod_{1\leq i<j\leq n}(X_j-X_i)(Y_i-Y_j)}
    {\displaystyle
      \prod_{i=1}^n\prod_{j=1}^n(X_i-Y_j)}
    \Perm(C),
\]
where $\Perm(C)$ denotes the permanent of $C$.
Moreover, $\Perm(C)$ is a nonzero rational function.
Consequently, $\Det C^{(2)}\neq0$, and $C^{(2)}$ has rank $n$ over
$\F(X,Y)$ for $\F\in\{\R,\C\}$.
\end{proposition}

\begin{lemma}[Mixed Hessian of the SI polynomial]
\label{lem:epsilon-si-hessian}
For every $n\geq2$ and $\varepsilon\neq0$, there exists a nonzero $(n-1)\times(n-1)$ minor
$M$ of $\HH_{X|Y}(\SI_{n,\varepsilon})$ such that
\[
  \SI_{n,\varepsilon}\nmid M.
\]
\end{lemma}

\begin{proof}
Since $\SI_{n,\varepsilon}=\SI_n-\varepsilon$, it has the same mixed Hessian as $\SI_n$.
We perform the following calculation in the rational function field $\F(X,Y)$. Clearing denominators then gives the corresponding polynomial identity.
We compute $\HH_{X|Y}(\SI_n)$ using logarithmic differentiation.
The first derivatives are:
\[
    \frac{\partial \SI_n}{\partial X_i} = \SI_n \sum_{b=1}^n \frac{1}{X_i - Y_b}, \quad \quad \frac{\partial \SI_n}{\partial Y_j} = -\SI_n \sum_{a=1}^n \frac{1}{X_a - Y_j}
\]
Taking the mixed second derivative and applying the product rule gives
\begin{align*}
    \frac{\partial^2 \SI_n}{\partial X_i \partial Y_j} &= \frac{\partial}{\partial X_i} \left( -\SI_n \sum_{a=1}^n \frac{1}{X_a - Y_j} \right) \\
    &= - \frac{\partial \SI_n}{\partial X_i} \left( \sum_{a=1}^n \frac{1}{X_a - Y_j} \right) - \SI_n \cdot \frac{\partial}{\partial X_i} \left( \sum_{a=1}^n \frac{1}{X_a - Y_j} \right) \\
    &= - \SI_n \left( \sum_{b=1}^n \frac{1}{X_i - Y_b} \right) \left( \sum_{a=1}^n \frac{1}{X_a - Y_j} \right) + \SI_n \frac{1}{(X_i - Y_j)^2}
\end{align*}

We can express this system over all $i, j \in [n]$ as a matrix equation.
Let $C^{(2)}$ be the $n \times n$ matrix with entries $C^{(2)}_{i,j} = \frac{1}{(X_i - Y_j)^2}$.
Let $u$ and $v$ be column vectors defined by $u_i = \sum_{b=1}^n \frac{1}{X_i - Y_b}$ and $v_j = \sum_{a=1}^n \frac{1}{X_a - Y_j}$.
We can now write the mixed Hessian as
\[
    \HH_{X|Y}(\SI_n) = \SI_n \cdot \left( C^{(2)} - u v^T \right).
\]
Thus $\rank\HH_{X|Y}(\SI_n) = \rank(C^{(2)} - u v^T)$.

By \cref{prop:borchardt}, $C^{(2)}$ has rank $n$, while $uv^T$ has rank at most one.
Hence
\[
  \rank\HH_{X|Y}(\SI_n)\geq n-1.
\]
Therefore $\HH_{X|Y}(\SI_{n,\varepsilon})$ has a nonzero
$(n-1)\times(n-1)$ minor $M$.
Every entry of this mixed Hessian is homogeneous of degree $n^2-2$, so
$M$ is homogeneous.
Since $\SI_{n,\varepsilon}$ is not homogeneous due to its nonzero constant
term, it cannot divide the nonzero homogeneous polynomial $M$.
\end{proof}

We obtain the probabilistic lower bound by the same framework used for the inner product set.

\begin{theorem}[Probabilistic lower bound for the $\varepsilon$-SI set]
\label{thm:epsilon-si-lower-bound}
For every $n\geq2$ and $\varepsilon>0$\footnote{For $n\in\{2,3\}$, these lower bounds are vacuous.},
\[
  \PCC_{\R}\left(S(\SI_{n,\varepsilon})\right)\geq2(n-3)
  \qquad\text{and}\qquad
  \PCC_{\C}\left(V(\SI_{n,\varepsilon})\right)\geq2(n-3).
\]
Consequently, the same lower bounds hold for the corresponding deterministic communication complexities.
\end{theorem}

\begin{proof}
By \cref{lem:epsilon-si-irreducible}, $\SI_{n,\varepsilon}$ is irreducible over both $\R$ and $\C$.

\begin{claim}\label{claim:border_inclusion}
$\partial S(\SI_{n,\varepsilon})=V(\SI_{n,\varepsilon})$.
\end{claim}

\begin{claimproof}
Continuity gives
$\partial S(\SI_{n,\varepsilon})\subseteq V(\SI_{n,\varepsilon})$.
Conversely, let $z\in V(\SI_{n,\varepsilon})$.
Since $\SI_n$ is homogeneous of degree $n^2$, Euler's theorem of homogeneous polynomials gives
\[
  \sum_{i=1}^n X_i\frac{\partial\SI_n}{\partial X_i}
  +\sum_{j=1}^n Y_j\frac{\partial\SI_n}{\partial Y_j}
  =n^2\SI_n.
\]
Then $\SI_n(z)=\varepsilon$, so evaluating Euler's identity at $z$
shows that $\nabla\SI_{n,\varepsilon}(z)\neq0$.
Moving a sufficiently small distance from $z$ in the two opposite
gradient directions therefore gives points where
$\SI_{n,\varepsilon}$ has opposite signs.
Hence $z\in\partial S(\SI_{n,\varepsilon})$, proving the reverse
inclusion.
\end{claimproof}

Let $u\in\R^n\times\R^n$ be given by
\[
  X_i=\varepsilon^{1/n^2},
  \qquad
  Y_i=0
  \qquad (i\in[n]).
\]
Then $\SI_n(u)=\varepsilon$, so $u\in V(\SI_{n,\varepsilon})$.
Moreover,
\[
  \frac{\partial\SI_{n,\varepsilon}}{\partial X_1}(u)
  =n\varepsilon^{1-1/n^2}\neq0,
\]
so $\nabla\SI_{n,\varepsilon}(u)\neq0$.
By \cref{prop:real_hypersurface},
\[
  \dim V(\SI_{n,\varepsilon})=2n-1.
\]
Together with the claim, this gives
\[
  \dim\left(\partial S(\SI_{n,\varepsilon})
  \cap V(\SI_{n,\varepsilon})\right)=2n-1.
\]
By \cref{lem:epsilon-si-hessian}, an $(n-1)\times(n-1)$ minor of
$\HH_{X|Y}(\SI_{n,\varepsilon})$ is not divisible by
$\SI_{n,\varepsilon}$.
Applying \cref{thm:framework-R-mixed-hessian} with $r=n-1$ gives
\[
  \PCC_{\R}\left(S(\SI_{n,\varepsilon})\right)\geq2(n-3).
\]

Over $\C$, \cref{prop:real_hypersurface} gives
$\dim V(\SI_{n,\varepsilon})=2n-1$.
The set $V(\SI_{n,\varepsilon})$ is Zariski closed, and hence equals its
Zariski closure.
Together with \cref{lem:epsilon-si-irreducible,lem:epsilon-si-hessian},
this verifies the hypotheses of
\cref{thm:framework-C-mixed-hessian} with $r=n-1$, which gives
\[
  \PCC_{\C}\left(V(\SI_{n,\varepsilon})\right)\geq2(n-3).
\]
\end{proof}

\section{Applications of communication lower bounds}
\label{sec:scanner}
In this section, we provide an application of our communication complexity lower bounds to an algorithmic setting. To the best of our knowledge, this  algorithmic model has not been considered in the literature before. But very similar models have been studied in different settings such as of one-pass algorithms, streams over $\N$, and turnstile models~\cite{AlonMatiasSzegedy, MuthukrishnanSurvey, LiNguyenWoodruff}.

Informally, the setup is as follows. The input is a string $(\alpha_1, \ldots, \alpha_n) \in \F^n$. The algorithm reads the string, after reading each
coordinate, it may write a polynomial function of that coordinate into its memory and optionally query a testing polynomial with the current contents of its memory. The decisions in the following steps are based on the result of the tests ($=, \neq$ in  the case of $\C$ and $<, >, =$ in the case of $\R$). The underlying algorithm may be deterministic or probabilistic. In the latter case, the probability of error is bounded by $1/3$. The formal definition follows.

\begin{definition}[Algebraic Scanner]
    Let $\F \in \{\R,\C\}$. An \Asc\ is an algorithm that, on input $\overline{\alpha} = (\alpha_1, \ldots, \alpha_n) \in \F^n$, scans the input from left to right. It has a memory register which is initially empty, that is, $M^{(0)} = \perp$. At each step $(i)$, it may perform the following operations.
    \begin{itemize}
        \item  Append the memory with a new value $f_i(\alpha_i)$: $M^{(i)} = \left(M^{(i-1)} , f_i(\alpha_i)\right)$, where $f_i \in \F[z]$. The memory stays unchanged, i.e., $M^{(i)} = M^{(i-1)}$, if nothing is appended in step $i$.
        \item Ask a tester query about the content of the memory at that step: Is $Q_i(M^{(i)}) \texttt{ op } 0?$, where $Q_{i} \in \F[z_1, \ldots, z_{\ell_i}]$, $\ell_i$ is the number of items in the memory and $\texttt{op} \in \{\neq, =\}$ if $\F = \C$ and $\texttt{op} \in \{<, >, =\}$, if $\F = \R$.
        \item At each step, at most $1$ memory append operation can happen and at most one query is performed. The result of the query may be used to decide the choice of polynomials $f$'s and $Q$'s in the following steps.
    \end{itemize}
    At the end of step $(n)$, the algorithm either accepts or rejects.

    We will say that $S \subseteq \F^n$ is accepted by an $\ltAsc$, if there is an \Asc\ that uses at most $\ell$ items in the memory and at most $t$ query polynomials and accepts $a \in \F^n$ if and only if $a \in S$. We say the cost of the scanner is $\ell+t$.

    A probabilistic scanner is defined as a finite-support distribution over deterministic scanners.
    We say that $S \subseteq \F^n$ is accepted by an $\ltAPsc$, if there is an \APsc\ that uses at most $\ell$ items in the memory and at most $t$ query polynomials. For $a \in \F^n$, if $a \in S$, it accepts $a$ with probability at least $2/3$ and, if $a \notin S$, it rejects with probability at least $2/3$.

\end{definition}

\begin{theorem}
    \label{thm:scanner-to-communication}
    Let $S \subseteq \F^n\times\F^n$ and let $\tilde{S} = \{X\circ Y \in \F^{2n} \mid (X,Y) \in S\}$, where $X\circ Y$ denotes concatenation of $X$ and $Y$. If $\PCC_{\F}(S) \geq r$, then any $\ltAPsc$ for $\tilde{S}$ is such that $\ell+t \geq r$.
\end{theorem}
\begin{proof}
    We prove the contrapositive. Suppose we have an $\ltAPsc$ for $\tilde{S}$, where $\ell+t < r$. Then given input $X$ to Alice and $Y$ to Bob, Alice runs the scanner on her input. When she finishes processing her input, Bob starts processing his input. Each stored value from Alice’s half becomes an Alice message; each stored value from Bob’s half becomes a Bob message; the referee simulates all tests.

    Notice that every time something is written in the memory, it becomes a possible input to the query polynomial at subsequent steps. Thus, the query polynomial receives at most $\ell$ inputs. Second, every time a test is performed, the protocol branches and increases the depth of the protocol. As there are at most $t$ tests and $\ell$ memory storage steps, we get that $r \geq \ell + t$, which contradicts our assumption.
\end{proof}

Using the above theorem and the lower bounds on the communication complexity from \Cref{sec:new-lower-bounds}, we obtain lower bounds on algebraic scanners for the following problems.

\begin{problems}
    {$\mathbf{\varepsilon}$-$\mathbf{\ell_2}$-distance}{vectors $\alpha, \beta \in \R^n$}{$\ell_2(\alpha, \beta) \leq \varepsilon$}
\end{problems}

\begin{problems}
    {bichromatic collision}{$n$ green points $\alpha \in \F^n$ followed by $n$ red points $\beta \in \F^n$}{Is there a pair of points of different colors that have the same value?}
\end{problems}

\begin{corollary}
    \label{cor:scanner-lb}
    Let $n \geq 2$ and let $\varepsilon >0$.  Any $\APsc$ over $\R$ for \textbf{$\mathbf{\varepsilon}$-$\mathbf{\ell_2}$-distance} has cost at least $2n$ and over $\F \in \{\R, \C\}$ for \textbf{bichromatic collision} has cost $\Omega(n)$.
\end{corollary}
\begin{proof}
    Note that for $(\alpha, \beta) \in \R^n \times \R^n$, $\ell_2(\alpha, \beta) \leq \varepsilon$ if and only if $(\alpha, \beta) \in S(\EQ_{n,\varepsilon^2})$. Also, note that $\alpha \in \F^n$ and $\beta \in \F^n$ have a bichromatic collision if and only if $(\alpha, \beta) \in V(\SI_{n})$
    Thus, using the probabilistic communication complexity lower bounds on $S(\EQ_{n,\varepsilon^2})$ and $V(\SI_{n})$ and by \Cref{thm:scanner-to-communication} we get the corollary.
\end{proof}

We also obtain lower bounds for two algebraic problems. To define these two problems, we introduce some notation. For the sake of this discussion, we will assume $\F = \C$.

\noindent \textbf{Notation.} A univariate monic\footnote{A univariate monic polynomial is a polynomial in one variable whose leading coefficient is $1$} polynomial of degree $n$ may be represented either by its coefficients or by its roots. For example, consider $f(t) = t^2 -4t + 3$. Here, we can represent it by the coefficient vector (1, -4, 3). But suppose we rewrite $f(t) = (t-1) (t-3)$, then we can represent it in terms of the multiset of roots, i.e., by $\{1, 3\}$. Here, we will assume the latter representation for univariate monic polynomials. We will denote this representation by $\textsf{roots}(f(t)) = \{\alpha_1, \ldots, \alpha_n\}$. This means, $f(t) = \prod_{i \in [n]} (t - \alpha_i)$.

\begin{problems}
{GCD}{$\textsf{roots}(f(t))$ followed by $\textsf{roots}(g(t))$, where $f(t)$ and $g(t)$ are univariate monic polynomials.}{Is gcd$(f(t), g(t)) \neq 1 ?$}
\end{problems}

\begin{problems}
    {Resultant}{$\textsf{roots}(f(t))$ followed by $\textsf{roots}(g(t))$, where $f(t)$ and $g(t)$ are univariate monic polynomials.}{Is Resultant$(f(t), g(t)) = \varepsilon?$}
\end{problems}

\begin{corollary}
    \label{cor:scanner-lb-gcd}
    Let $n \geq 2$ and $\varepsilon > 0$.  Any $\APsc$ over $\C$ for both GCD and Resultant has cost $\Omega(n)$.
\end{corollary}
\begin{proof}
    Note that resultant of $f(t)$ and $g(t)$ is equal to $$\prod_{i\in [n]} \prod_{j \in [n]} (\alpha_i - \beta_j).$$ And this is $0$ if and only if their gcd is non-trivial, i.e., it is not equal to $1$~\cite[Chapter 3, Section 1]{cox2005using}.

    Thus, by lower bounds on the communication complexity of $V(\SI_{n})$ and $V(\SI_{n, \varepsilon})$ and by \Cref{thm:scanner-to-communication}, we get the corollary.
\end{proof}
\section{Extension to bounded-time BSS machines}
\label{sec:bss}

In this section we generalize the communication model and prove that our framework for proving probabilistic lower bounds can be adapted in this general model. Our more general model is inspired by the well-known Blum--Shub--Smale model of computation~\cite{BCSS98}.

\subsection{Bounded-time BSS machines}

In this general model, Alice, Bob, and the referee are allowed to compute more general classes of functions. Implicitly, they are allowed to use computation defined by the following computational model:
\begin{definition}[Bounded-time BSS machine]
    A \emph{bounded-time BSS machine} is a finite computation tree consisting of
    \begin{itemize}
        \item \textbf{Input nodes} that read input from $\R$ into specified registers, $(R_1,...,R_n)$, together with finitely many additional registers initialized to fixed real constants
        \item \textbf{Computation nodes} that execute $R_i \leftarrow R_j * R_k$, where $* \in \{+,-,\times,/\}$, where division by zero is disallowed.
        \item \textbf{Branch nodes} that branch according to $R_i >0$ or $R_i =0$ or $R_i <0 $.
        \item \textbf{Output leaves} that outputs a single real-valued number.
    \end{itemize}
\end{definition}

\begin{remark}
    Usually, in a standard BSS model, the output leaf writes out the (vector-valued) contents of the registers. However, since the communication model we consider is single-valued, we will only be considering BSS machines with single-valued output.

    Note that BSS machines are usually directed graphs, not trees, but by assuming every input halts within a uniform finite time bound we may unravel the BSS machine to a (potentially exponential sized) tree as defined above.
\end{remark}

Now if we just follow one path, $\tau$ in this tree then by regarding the
inputs as variables, the output is a rational function, which we denote as $M^\tau(X) \in \R(X)$. Hence, globally, the output of the BSS machine, $M(X)$, is just a piecewise rational function. We assume these rational function have \emph{no} poles on their corresponding pieces.

Furthermore, the set of points which travel along $\tau$ is given by a semialgebraic set. Indeed, each branching along $\tau$ is given by the sign of a rational function which can be described by (in)equalities of polynomials.\footnote{Since if $Q(X) = A(X)/B(X)$ then $ Q(X) = 0 \iff A(X) = 0$ and $ Q(X) > 0 \iff A(X)B(X) > 0$ and $ Q(X) < 0 \iff A(X)B(X) < 0$. (Assuming $B(X) \neq 0$).}

Thus, a bounded-time BSS machine induces a finite partition of $\R^n$ into semialgebraic sets, one for each feasible computation path, and on each such set its output agrees with a rational function.

\subsection{Extension of the full-rank mixed Hessian criterion to rational functions}
\label{sec:full-rank-hessian-ext}

We start by generalizing \cref{lem:full-rank-criterion} to work for rational functions. This will play a critical role in \cref{thm:framework-R-BSS}.

\begin{lemma}[Fractional full-rank criterion for powers]\label{lem:full-rank-criteria-fractional}
    Let $\F$ be a field of characteristic zero, let $p \in \F[X,Y]$ be irreducible, and define determinant of $(n+1) \times (n+1)$ matrix
    $$
    R_p \quad \coloneqq \quad \Det\begin{pmatrix}
        \HH_{X|Y}(p) & \nabla_X p \\ (\nabla_Y p)^T & 0
    \end{pmatrix}.
    $$
    Suppose that $p \nmid R_p$. Then for every
    $$
    g = p^m\frac{a}{b}, \qquad m \geq 2, \qquad p \nmid a, \qquad p \nmid b
    $$
    with $a,b \in \F[X,Y]\setminus \{0\}$, we have
    $$
    \rank \HH_{X|Y}(g) = n.
    $$
\end{lemma}

\begin{proof}
    Let
    \[
        A=\F[X,Y]_{\langle p\rangle}
        =\left\{\frac{c}{d}:c,d\in\F[X,Y],\ d\neq0,\ p\nmid d\right\},
        \qquad h=\frac{a}{b}.
    \]
    First, $h$ is a unit in $A$, since $p\nmid a$ and hence $h^{-1}=b/a\in A$.
    Second, $A$ is closed under differentiation: if $r/s\in A$, then
    \[
        \frac{\partial}{\partial Z}\left(\frac{r}{s}\right)
        =\frac{s\,\partial_Zr-r\,\partial_Zs}{s^2}\in A,
    \]
    since $p\nmid s^2$.
    Finally, since $\F[X,Y]$ is a unique factorization domain and $p$ is irreducible, $\langle p\rangle$ is prime. Hence, $A$ is a local domain with maximal ideal $pA$, so $A/pA$ is a field. The images of both $h$ and $R_p$ in this field are nonzero: the former because $h$ is a unit and the latter because $p\nmid R_p$.

    These are precisely the properties of $h$ and reduction modulo $p$ used in the proof of \cref{lem:full-rank-criterion}. Hence, the same calculation applies over $A$ and gives
    \[
        \Det\HH_{X|Y}(g)
        =p^{n(m-2)+n-1}Q,
        \qquad
        Q\equiv-m^n(m-1)h^nR_p\not\equiv0\pmod{pA}.
    \]
    Here, the nonvanishing follows from $m\geq2$, $\operatorname{char}(\F)=0$, and the preceding observations about $h$ and $R_p$.
    Therefore, $Q\neq0$ and $\Det\HH_{X|Y}(g)\neq0$. Since the fraction field of $A$ is $\F(X,Y)$, it follows that
    \[
        \rank_{\F(X,Y)}\HH_{X|Y}(g)=n.
    \]
\end{proof}

\begin{remark}
    The above is a \emph{strict} generalization of \cref{lem:full-rank-criterion}, since we retrieve the result by simply setting $b = 1$.
\end{remark}

\subsection{Framework for BSS probabilistic lower bounds}
We now consider the extension of \cref{def:set-recognition-protocol} and \cref{def:probabilistic-set-recognition-protocol} given by replacing every polynomial with bounded-time BSS machine. I.e. the protocol is exactly the same but the local computations and the branching are now done by bounded-time BSS machines. Furthermore, the communication complexity, $c(g)$, in this new model is defined exactly as in \cref{def:comm-complexity} but with $a_1,...,a_{r_1}$, $b_1,...,b_{r_2}$ and $P$ and $g$ all given by bounded-time BSS machines.

Surprisingly, it turns out that even with this general model we can get the exact same framework for lower bounds of semialgebraic set-recognition! This result seems to hint that increasing the local power of the model does not result in greater computational power for set-recognition of semialgebraic sets.

Before we state and prove the main theorem, we need a generalization of \cref{lem:comm-matrix}, since path-specific computations of bounded-time BSS machines are rational functions rather than polynomial functions. This generalizes \cite{Gri2008} whose result is inspired by the much more general result, \cite[Theorem~2]{Abelson80}.

\begin{lemma}[Mixed Hessian bound for rational computations]\label{lem:comm-matrix-rat}
    Let $\F \in \{\R, \C\}$, and let $g \in \F(X,Y)$ be a rational function. Suppose that
    $$
    g(X,Y) = Q(a_1(X),...,a_{r_1}(X),b_1(Y),...,b_{r_2}(Y)),
    $$
    is a minimum-size representation of $g$, where $a_1,...,a_{r_1} \in \F(X)$, $b_1,...,b_{r_2} \in \F(Y)$ and $Q \in \F(A_1,...,A_{r_1},B_1,...,B_{r_2})$ are all rational functions. Then $\rank \HH_{X|Y}(g) \leq \min\{r_1,r_2\}$ and consequently
    $$
    c(g)=r_1+r_2 \geq 2 \cdot \rank \HH_{X|Y}(g),
    $$
    where the rank is taken over $\F(X,Y)$.
\end{lemma}
\begin{proof}
    In the following all derivatives are formal derivatives and the calculations can be done over any open subset of $\F^n \times \F^n$ on which the denominators are non-zero.

    Let $a = (a_1,...,a_{r_1})$ and $b=(b_1,...,b_{r_2})$. All the considered functions at least of regularity $\mathcal{C}^2$ (as we are dealing with quotients of polynomials), so Schwartz's theorem can be applied (and the partial derivatives can be swapped).
    By the same calculation as in the proof of \cref{lem:comm-matrix} we have that
    $$
    \HH_{X|Y}(g) = J_X(a) \HH_{A|B}(Q)(a,b) J_Y(b)^T,
    $$
    where $J_X(a)$ is the $n \times r_1$ Jacobian matrix and $J_Y(b)$ is the $n \times r_2$ Jacobian matrix. Thus
    $$
    \rank \HH_{X|Y}(g) \leq \min\{r_1,r_2\}.
    $$
    and so
    $$
    c(g) = r_1+r_2 \geq 2 \rank \HH_{X|Y}(g)
    $$
    as desired.
\end{proof}

\begin{theorem}\label{thm:framework-R-BSS}
    Let $S \subseteq \R^n\times\R^n$ be a semialgebraic set and let $F \in \R[X,Y]$ satisfy
    \begin{enumerate}
        \item $F$ is irreducible over $\R$
        \item $\partial S \cap V(F)$ has dimension $2n-1$
        \item $F \nmid R_F$
    \end{enumerate}
    Then $S$ has BSS-probabilistic communication complexity at least $2n$.
\end{theorem}

This proof is a lift of the proof of \cref{thm:framework-R}, so while it can mostly stand on its own, it is essentially simply a modification of \cref{thm:framework-R}.

\begin{proof}
    Throughout this proof we only consider paths through trees which are realizable; i.e. paths for which there exists some input that follows the given path. Furthermore, throughout this proof, all fractions are in reduced form.

    Let $\mathcal{C}$ be a probabilistic communication protocol recognizing $S$.

    Note that a protocol $C$ consists of a computation tree of computation trees: Indeed, we have an outer tree which branches based on the results of BSS machines which themselves are computation trees.

    Thus we define a \emph{trace}, $\tau$, of a protocol $C$ to consist of
    \begin{itemize}
        \item A path through the outer tree.
        \item A path through every communication BSS machine encountered on that outer path. This could be either Alice's machine or Bob's machine.
        \item A path through every transcript-testing BSS machine encountered on that outer path. This is the referee's BSS machine.
    \end{itemize}
    and we denote by $\operatorname{Tr}(C)$ the set of traces of $C$.

    Along such a $\tau$, the communicated values are rational functions and all testing functions are conducted by rational functions. Therefore, we see that any branch test occurring along $\tau$
    has the form
    $$
        G_{C,\tau,i}(X,Y) = \frac{N_{C,\tau,i}(X,Y)}{D_{C,\tau,i}(X,Y)}
    $$
    where $i$ indexes the branch tests.

    Then, we define (assuming without loss of generality that $N_{C,\tau,i} \not\equiv 0$):
    $$
    P_{C,\tau} = \prod_i N_{C,\tau,i}
    $$
    which is a finite product, since every path is finite. Furthermore, the set
    $$
    \mathcal{P} = \{P_{C,\tau} \mid C \in \mathcal{C}, \tau \in \operatorname{Tr}(C)\}
    $$
    is finite since there are only finitely many protocols in $\mathcal{C}$ and since we consider only \emph{bounded-time} BSS machines, there are only finitely many traces, $\tau$, in $C$.

    \begin{claim}\label{claim:subset-bss-case}
        $$
        \partial S \subseteq \bigcup_{C \in \mathcal{C}} \bigcup_{\tau \in \operatorname{Tr}(C)} V(P_{C,\tau}) =: \mathcal{V}
        $$
    \end{claim}
    \begin{claimproof}
        Suppose, for contradiction, that there exists $u \in \partial S$ such that $u \notin \mathcal{V}$. So $P_{C, \tau}(u) \neq 0$ for all $P_{C,\tau} \in \mathcal{P}$.%

        For each $C \in \mathcal{C}$ we denote by $\tau_C$ the unique trace followed by $u$. Since we disallow division by zero in BSS machines, every denominator along the trace encountered is nonzero. Moreover, since $u \notin \mathcal{V}$ every numerator encountered is also nonzero. Since the nonzero set of these numerators and denominators is an open set it follows that there exists an open neighborhood $B_C$ around $u$. Moreover, by continuity the signs of the numerators and denominators are unchanged on $B_C$.
        Thus every element of $B_C$ follows the trace, $\tau_C$, and outputs the same value as $u$. Hence, the \emph{finite} intersection of open neighborhoods,
        $$
        B_u = \bigcap_{C \in \mathcal{C}} B_C,
        $$
        is itself an open neighborhood of $u$ on which the protocol, $\mathcal{C}$, has constant output.

        More precisely this observation shows that, on $B_u$, we have
        $$
        \Pr[\mathcal{C} \text{ accepts } z] =  \Pr[\mathcal{C} \text{ accepts } u]
        $$
        and so is constant on $B_u$.

        But, $u$ is on the boundary of $S$ and so $B_u$ must contain two points $z_{\text{in}} \in S$ and $z_{\text{out}} \notin S$ and by definition we must have $\Pr[\mathcal{C} \text{ accepts } z_{\text{in}}] \geq 2/3$ and $\Pr[\mathcal{C} \text{ accepts } z_{\text{out}}] \leq 1/3$. This gives us the desired contradiction:
        $$
        2/3 \leq \Pr[\mathcal{C} \text{ accepts } z_{\text{in}}] =  \Pr[\mathcal{C} \text{ accepts } z_{\text{out}}] \leq 1/3
        $$
        since $z_{\text{in}},z_{\text{out}} \in B_u$ and $\Pr[\mathcal{C} \text{ accepts } z]$ is constant on $B_u$.
    \end{claimproof}

    Assume now for contradiction that $F\nmid P_{C,\tau}$ for all $C$ and $\tau$. By \Cref{claim:subset-bss-case} and by \Cref{fact:dim-red} we see that
    $$
    \dim(\partial S \cap V(F)) \leq  \dim\left(\bigcup_{C \in \mathcal{C}} \bigcup_{\tau \in \operatorname{Tr}(C)} (V(P_{C,\tau}) \cap V(F))\right) \leq 2n-2
    $$
    which contradicts assumption 2.

    Thus, $F$ divides some $P_{C_0,\tau_0}$ and since $F$ is irreducible (in particular prime) $F$ must divide some numerator, $N_{C_0,\tau_0,0}$.

    \begin{claim}\label{claim:transcript-origin}
        $N_{C_0,\tau_0,0}$ comes from a transcript test. (In particular not from any private computation done by Alice or Bob).
    \end{claim}
    \begin{claimproof}
        Assume it came from a private computation, say from Alice. Then $N_{C_0,\tau_0,0}$ would only depend on $X$:
        $$
        N_{C_0,\tau_0,0}(X,Y) = A(X).
        $$
        But then $F(X,Y) \mid A(X)$ and so $F$ would only depend on $X$ and be independent of $Y$. This would imply that $\nabla_YF = 0$ and in particular $R_F = 0$ which means $F \mid R_F$. This contradicts $F \nmid R_F$ and so $N_{C_0,\tau_0,0}$ must come from a transcript test.
    \end{claimproof}

    Let $Q(X,Y)$ be the rational function such that (written in reduced representation)
    $$
    Q(X,Y) = \frac{N_{C_0,\tau_0,0}(X,Y)}{D_{C_0,\tau_0,0}(X,Y)}
    $$
    What we have shown so far is that by \Cref{claim:transcript-origin}, $Q(X,Y)$ is the result of a transcript test and so
    $$
    Q(X,Y) = \hat{Q}(a_1(X),...,a_{r_A}(X),b_1(Y),...,b_{r_B}(Y)) = \frac{N_{C_0,\tau_0,0}(X,Y)}{D_{C_0,\tau_0,0}(X,Y)}
    $$
    with
    $$
    F \mid N_{C_0,\tau_0,0}
    $$
    and since $Q$ was written in reduced form we have $\gcd(N_{C_0,\tau_0,0}, D_{C_0,\tau_0,0}) = 1$ so we also have
    $$
    F \nmid D_{C_0,\tau_0,0}
    $$

    We may therefore write $N_{C_0,\tau_0,0} = F^m \hat{N}$ with $m \geq 1$ and $F \nmid \hat{N}$. This, in turn, gives
    $$
    Q^2 = F^{2m}\left(\frac{\hat{N}}{D_{C_0,\tau_0,0}}\right)^2
    $$
    with $2m \geq 2$ and $F \nmid \hat{N}^2$ and $F \nmid D_{C_0,\tau_0,0}^2$. Thus, we now apply the fractional full-rank criterion, \cref{lem:full-rank-criteria-fractional}, to get that
    $$
    c(Q^2) \geq 2 \rank(\HH_{X|Y}(Q^2)) = 2n
    $$
    where the first inequality is from \cref{lem:comm-matrix-rat}. Finally, by the same observation in the proof of \cref{thm:framework-R} we get
    $$
    c(Q) \geq c(Q^2) \geq  2n
    $$
    as desired.
\end{proof}

\begin{remark}
    One can define \emph{complex} bounded-time BSS machines, that is, BSS machines when the underlying field $\F$ is $\C$,  wherein we branch only on $=$ or $\neq$. With this model one should be able to get a BSS analogue of the complex framework \cref{thm:framework-C}.

    Indeed, if we consider all the rational functions, $T_1,...,T_m$ along the trace in following $\neq0$ and take their product to be $G_k$. Then, in reduced form, we would have that
    $$
    G_k = \frac{N_k}{D_k}
    $$
    and, by the same argument as \Cref{clm:reject-path} and \Cref{claim:F-div-Gk}, we may conclude that
    $$
    F \mid N_k \qquad \text{ and } \qquad F \nmid D_k.
    $$
    Furthermore, by the same argument as in \Cref{claim:transcript-origin} it must be the case that a factor of $N_k$ comes from the numerator of a transcript test.

    The proof then follows exactly as in \cref{thm:framework-R-BSS} to conclude that
    $$
    c(G_k) \geq c(G_k^2) \geq 2n,
    $$
    i.e. the depth of the protocol, which is at least $c(Q)$,  is at least $2n$.
\end{remark}
\section*{Acknowledgements}
\addcontentsline{toc}{section}{Acknowledgements}

We thank the organizers and participants of the Bellairs Complexity Workshop 2026 in Barbados and the WAVE Workshop 2025 in Copenhagen for the stimulating environments and discussions that helped shape this work.

\printbibliography[heading=bibintoc]

\appendix

\section{Proof of the mixed Hessian rank bound}
\label{app:comm-matrix-proof}

\getkeytheorem{commmatrixlemma}
\begin{proof}

    Let $i,j\in \{1, \cdots, n\}$. We use the same notations as in  \Cref{def:comm-complexity}. In the rest of the proof, all the functions are at least of regularity $\mathcal{C}^2$ (as we are dealing with polynomials), so Schwarz's theorem can be applied (and the partial derivatives can be swapped).
    We compute $\HH_{X|Y}$:
    \begin{align*}
        \frac{\partial g}{\partial X_i}(X,Y) &= \frac{\partial Q(a_1(X), \cdots, a_{r_1}(X), b_1(Y), \cdots, b_{r_2}(Y))}{\partial X_i}(X,Y)\\
        \text{(chain rule)\quad} &= \sum_{p=1}^{r_1} \frac{\partial Q(a_1(X), \cdots, a_{r_1}(X), b_1(Y), \cdots, b_{r_2}(Y))}{\partial a_p(X)} \frac{\partial a_p(X)}{\partial X_i}
    \end{align*}
    Hence,
    \begin{align*}
        \frac{\partial^2 g}{\partial X_i \partial Y_j}(X,Y) &= \sum_{p=1}^{r_1} \frac{\partial^2 Q(a_1(X), \cdots, a_{r_1}(X), b_1(Y), \cdots, b_{r_2}(Y))}{\partial a_p(X) \partial Y_j} \frac{\partial a_p(X)}{\partial X_i}\\
        & \quad \quad + \frac{\partial Q(a_1(X), \cdots, a_{r_1}(X), b_1(Y), \cdots, b_{r_2}(Y))}{\partial a_p(X)}\underbrace{\frac{\partial^2 a_p(X) }{\partial X_i \partial Y_j}}_{= 0}\\
        \text{(chain rule)\quad}&= \sum_{p=1}^{r_1} \sum_{q=1}^{r_2} \left(\frac{\partial b_q(Y)}{\partial Y_j}\right)\left(\frac{\partial^2 Q(a_1(X), \cdots, a_{r_1}(X), b_1(Y), \cdots, b_{r_2}(Y)) }{\partial a_p(X) \partial b_q(Y)}\right)\left(\frac{\partial a_p(X)}{\partial X_i}\right)\\
    \end{align*}

    Thus, $$ \HH_{X|Y}(g) =  \begin{pmatrix}
            \partial a_1(X)/ \partial X \\
            \colon \\
            \partial a_{r_1}(X)/ \partial X
\end{pmatrix} \HH_{a|b}(Q)
\begin{pmatrix}
            \partial b_1(Y)/ \partial Y \\
            \colon \\
            \partial b_{r_2}(Y)/ \partial Y
\end{pmatrix}^T
$$

where $a=(a_1, \cdots, a_{r_1})$ and $b=(b_1, \cdots, b_{r_2})$.

Thus, $\rank(\HH_{X|Y}(g)) \leq \min(r_1, r_2)$ and $r_1 + r_2 \geq 2 \min(r_1, r_2) \geq 2\rank(\HH_{X|Y}(g)) $. We conclude.

\end{proof}
\end{document}